\documentclass{scrartcl}
\usepackage[utf8]{inputenc}
\usepackage{authblk}
\usepackage[pagebackref]{hyperref}
\usepackage{amsthm}

\expandafter\let\csname equation*\endcsname\relax
\expandafter\let\csname endequation*\endcsname\relax
\usepackage{amsmath}
\usepackage{amssymb}
\usepackage{dsfont}
\usepackage{comment}
\usepackage{amsfonts}
\usepackage{needspace}
\usepackage{wrapfig}            

\usepackage[textsize=scriptsize]{todonotes}
\usepackage{marginnote}

\newcommand{\revtodo}[2][]{{%
 \let\marginpar\marginnote
 \reversemarginpar
 \renewcommand{\baselinestretch}{0.8}%
 \todo[#1]{#2}}}

\usepackage{booktabs}
\usepackage{multirow}

\usepackage{caption}

\usepackage{tikz}
\usetikzlibrary{positioning,backgrounds,patterns,calc}
\usetikzlibrary{arrows,shapes,intersections,fit,decorations.pathreplacing,scopes}
\usetikzlibrary{decorations.pathmorphing}
\usetikzlibrary{automata}

\pgfdeclarelayer{bg}
\pgfsetlayers{bg,main}

\tikzstyle{vert}=[circle,inner sep=1.5,fill=white,draw,minimum size=.5cm]
\tikzstyle{vertex}=[circle,draw=black,minimum size=8pt,inner sep=1pt]
\tikzstyle{vertex2}=[circle,draw=black,minimum size=15pt,inner sep=2pt]
\tikzstyle{edge}=[]
\tikzstyle{ypath}=[ultra thick]
\tikzstyle{dottedEdge}=[dotted,thick]
\tikzstyle{small-vertex}=[circle,draw=black,minimum size=6pt,inner sep=0pt,fill=white]
\tikzstyle{thinedges}=[draw=gray!30]
 \tikzstyle{boxes}=[draw,thick, rounded corners=3mm,text width=2.7cm,align=center,text opacity=1,fill opacity=1,fill=white]
\tikzstyle{unk}=[fill=gray!25!white]

\usepackage{bbding}
\usepackage{placeins}

\usepackage{needspace}
\usepackage{etoolbox}

{\bfseries}{\normalfont}
\newtheorem{theorem}{Theorem}{\bfseries}{\normalfont}
\newtheorem{obs}{Observation}{\bfseries}{\normalfont}
\newtheorem{cor}{Corollary}{\bfseries}{\normalfont}
\AtBeginEnvironment{cor}{\Needspace{2\baselineskip}}
\newtheorem{prop}{Proposition}{\bfseries}{\normalfont}
\AtBeginEnvironment{prop}{\Needspace{2\baselineskip}}
\newtheorem{lem}{Lemma}{\bfseries}{\normalfont}

\newtheorem{defi}{Definition}{\bfseries}{\normalfont}
\AtBeginEnvironment{defi}{\Needspace{2\baselineskip}}

\DeclareMathOperator{\poly}{poly}

\newcommand{\decprob}[3]{%
  \begin{center}%
    \begin{minipage}{0.9\linewidth}%
      \textsc{#1}\\
      \textbf{Input:} #2\\
      \textbf{Question:} #3
    \end{minipage}%
  \end{center}%
}

\newcommand{\proofparagraph}[1]{
  \medskip
  \noindent\emph{#1}\quad}

\usepackage{units}

\usepackage[numbers,sort]{natbib}

\usepackage{paralist}

\usepackage{enumitem}
\usepackage[capitalize]{cleveref}

\usepackage{xspace}

\newlist{propertylist}{enumerate}{10}
\setlist[propertylist]{label=\roman*),leftmargin=*}
\crefname{propertylisti}{Property}{Properties}

\newcommand{\SL}{\textsc{Sub\-graph}\xspace}
\newcommand{\ML}{\textsc{ML-Sub\-graph}\xspace}
\newcommand{\MLlong}{\textsc{Multi-Layer Sub\-graph}\xspace}
\newcommand{\PMLlong}{$\Pi$ \MLlong}
\newcommand{\PML}{$\Pi$\nobreakdash-\ML}

\newcommand{\gprop}[1]{\textsf{#1}}
\newcommand{\GML}[1]{\gprop{#1}-\ML}

\newcommand{\MatchML}{\GML{Matching}}
\newcommand{\FactorML}{\GML{$c$-Factor}}
\newcommand{\PathML}{\GML{Hamiltonian}}

\newcommand{\MCC}{\textsc{Multicolored Clique}\xspace}

\newcommand{\NP}{\ensuremath{\textsf{NP}}\xspace}
\newcommand{\XP}{\ensuremath{\textsf{XP}}\xspace}
\newcommand{\FPT}{\ensuremath{\textsf{FPT}}\xspace}
\newcommand{\W}[1]{\ensuremath{\textsf{W[#1]}}\xspace}

\newcommand{\paramProblem}{\ensuremath{L}}

\crefname{rrule}{Rule}{Rules}

\newcommand{\NN}{\mathbb{N}\xspace}
\newcommand{\F}{\mathcal{F}\xspace}

\graphicspath{{images/}}

\title{Assessing the Computational Complexity of Multi-Layer Subgraph Detection\footnote{An extended abstract~\cite{BKKMNS2017} appeared in the proceedings of the 10th International Conference on Algorithms and Complexity (CIAC 2017), held in Athens, Greece, May 24-26 2017. This long version appeared in Network Science~\cite{BredereckKKMNS19}. Work started while all authors were with TU~Berlin.}}

\begin{document}



\author[1]{Robert Bredereck}
\affil[1]{\small Algorithmics and Computational Complexity, Faculty~IV, 
 TU Berlin, Germany,
 \texttt{\{robert.bredereck,~h.molter,~rolf.niedermeier\}@tu-berlin.de}}
\author[2]{Christian Komusiewicz}
\affil[2]{Fachbereich Mathematik und Informatik, Philipps-Universität Marburg, Germany, 
 \texttt{komusiewicz@informatik.uni-marburg.de}}
\author[3]{Stefan Kratsch}
\affil[3]{Humboldt-Universit\"at zu Berlin, Germany, 
 \texttt{kratsch@informatik.hu-berlin.de}}
\author[1]{Hendrik~Molter}
\author[1]{Rolf~Niedermeier}
\author[4]{Manuel~Sorge}

\affil[4]{Faculty of Mathematics, Informatics and Mechanics, University of Warsaw, Poland, 
  \texttt{manuel.sorge@mimuw.edu.pl}
  }

\date{}

\maketitle


\begin{abstract}
Multi-layer graphs consist of several graphs, called layers, where
the vertex set of all layers is the same but each
layer has an individual edge set.
They are motivated by real-world problems where entities (vertices) 
are associated via multiple types of relationships (edges in different layers).
We chart the border of computational (in)tractability 
for the class of subgraph detection problems on multi-layer graphs,
including fundamental problems such as maximum-cardinality matching, finding certain
clique relaxations, or path problems.
Mostly encountering hardness results, sometimes even for two or three layers,
we can also spot some islands of computational tractability.


\bigskip


\end{abstract}


\section{Introduction}
\emph{Multi-layer} graphs consist of several graphs, called \emph{layers}, where
the vertex set of all layers is the same but each
layer has an individual edge set~\cite{MR11,Kiv+14,Boc+14}. They are also known as multi-dimensional networks~\cite{BCGM13}, multiplex networks~\cite{MRMMO10},
edge-colored multigraphs~\cite{CY14,ALMS15,agrawal_et_al18}, among others~\cite{Kiv+14}. In recent years, multi-layer
graphs have gained considerable attention 
because observational data often comes in
a multimodal nature. Typical topics studied here include
clustering~\cite{BGHS12,DFVN12,DFVN14,JP09,chen2017parameterized}, 
detection of network communities~\cite{KL15,ZWZK07}, data privacy~\cite{RMT15}, matching problems~\cite{chen2018stable}, and general
network properties~\cite{BCGM13}.

In several of these applications, the goal is to identify vertex subsets
of a multi-layer graph that exhibit a certain structure in each
layer. For example, motivated by applications in genome
comparison in computational biology, \citet{GHPR03} searched for maximal vertex subsets
in a two-layer graph that induce a connected graph in each
layer.
\citet{JP09} and \citet{BGHS12} searched for vertex
subsets that induce dense subgraphs in many layers. Such vertex subsets model
communities in a multimodal social network.

To the best of our knowledge, however, a systematic study on
computational complexity classification of such problems is lacking.
Typically, authors observe the generalization of hardness results for
the one-layer case to the multi-layer one~\cite{BGHS12,JP09} or
perform case studies for special types of subgraphs~\cite{ALMS15,agrawal_et_al18}.
From a general algorithmic viewpoint we are aware only of Cai and
\citet{CY14} who derived some complexity results for general classes of subgraphs, focusing mostly
on the two-layer case. Our aim in this article is hence to build a
general foundation for studying the worst-case computational
complexity of a wide class of multi-layer subgraph problems, and to
provide results that may stimulate more specific algorithmic analyses.

We first give a general problem definition that encompasses the
problems sketched above. All of these problems can be phrased as finding a
large vertex subset that induces graphs with an interesting property
in many layers. 
Motivated by the heterogeneity of the desired
properties and analogous to the one-layer case~\cite{GJ79}, we give a problem definition with a generic 
\emph{graph property~$\Pi$} as part of the problem name. Here, $\Pi$ is formally a fixed set of graphs. 
For example, $\Pi$ may be the
set of all connected graphs or the set of all complete graphs. This allows us to study a wide variety of graph properties; 
we remark that throughout the paper we always assume that testing membership is decidable for the graph properties $\Pi$ under consideration.

\decprob{\PMLlong (\PML)}%
{A set of undirected graphs $G_1,\ldots,G_t$ all on the same vertex set~$V$ and two
positive integers~$k$ and~$\ell$.} 
{Is there a vertex set~$X\subseteq V$ with $|X|\geq k$ such that for at
least~$\ell$ of the input graphs~$G_i$ it holds that the induced subgraph~$G_i[X]$ has the property~$\Pi$?}

\noindent We study \PML mostly in the context of parameterized
computational complexity analysis~\cite{Cy15,DF13,FG06,Nie06}.  In contrast to classical
computational complexity analysis, where mostly only the input size is
used to estimate the worst case running time of an algorithm, we
consider running times depending also on secondary measures of the
input instances, so-called \emph{parameters}. We study \PML with
respect to the following natural parameters: 
\begin{compactitem}
\item the total number~$t$ of
layers,
\item the number~$k$ of vertices to select, and
\item the
number~$\ell$ of layers to select, as well as
\item their dual parameters~$|V| - k$ and $t - \ell$,
\end{compactitem} 
and combinations of
these parameters. A combination of two parameters can be thought of as
the parameter obtained by taking the sum of the two individual
parameters. The parameters~$|V| - k$ and $t - \ell$ are also
called \emph{deletion} parameters because they correspond to the
equivalent problem of deleting at most $|V| - k$ vertices and
$t - \ell$ layers to obtain a desired subgraph in each remaining
layer.

In analyzing the parameterized computational complexity with respect
to a parameter~$p$, we aim to find efficient algorithms if $p$ is
small. This is formalized in the concept of \emph{fixed-parameter
  tractability (\FPT)} which states that there exists an algorithm that
produces a solution in $f(p) \cdot \poly(|I|)$ time, where~$f$ is a
computable function and $|I|$ is the size of the input instance. In
contrast, polynomial-time algorithms for a constant parameter
value~$p$ may have running time $O(|I|^{g(p)})$ for some computable
function~$g$. Such algorithms, called \emph{\XP-algorithms},
usually have a prohibitively large running time, even for small values of~$p$. Using \emph{W[1]-hardness}, a
concept analogous to \NP-hardness, one can show that a problem is unlikely to admit an \FPT-algorithm.
To fully assess the computational complexity behavior of problems it is important to study them under various parameterizations and combinations thereof~\cite{DF13,Nie10,KN12}.
Note that \W{1}-hardness when parameterized by a
combined parameter~$(p, q)$ implies \W{1}-hardness when parameterized
by either~$p$ or~$q$. On the contrary, \W{1}-hardness with respect to a
parameter~$p$ may still allow for an \FPT-algorithm with respect to a combination
of~$p$ with an additional parameter. 

\begin{table}[t]
    
    \caption{Result Overview; $k$ is the number of vertices to select, $\ell$ is the
    number of layers to select, and $t$ is the total number of layers. 
    A graph property $\Pi$ is \emph{vertex-partitionable} if one can compute
    partitions of a graph into maximal components that each satisfy $\Pi$ in
    polynomial time; for details see
    \autoref{def:partitioning-properties} in \autoref{sec:meta}. A graph property is \emph{staggered} if we can build a certain gadget based on the property, see \autoref{def:staggered} in \autoref{sec:meta} for details. For all \FPT, \XP, and \W{1}-hardness results, we also have corresponding \NP-hardness results. (In the case of vertex-partitionable graph properties, we get \NP{}-hardness if the property is also staggered, which is the case for all properties we consider in this paper.)} 
\label{tab:results}
  \centering
  \begin{tabular}{@{}l@{}l@{}l@{}}
    \toprule
    Graph Property $\Pi$ & Complexity of \PML{} & Reference\\ 
    \midrule
    Hereditary and & & \\
    - contains finitely many graphs & polynomial-time solvable &
    \autoref{thm:hereditary}\\
    - includes all complete \emph{and} & \multirow{2}{*}{\FPT wrt.\ $(k, \ell)$} &  \multirow{2}{*}{\autoref{thm:hereditary}}\\ \quad all edgeless graphs & &\\
    - includes \emph{either} all complete & \multirow{2}{*}{\W{1}-hard wrt.\ $k$ for all $\ell$} & \multirow{2}{*}{\autoref{thm:hereditary}} \\ \quad \emph{or} all edgeless graphs & &\\
    - characterizable by finitely many & \multirow{2}{*}{\FPT wrt.\ $(|V|-k, t-\ell)$} & \multirow{2}{*}{\autoref{prop:forb-subgr}}\\ \quad forbidden induced subgraphs & &\\
    \midrule 
    Vertex-partitionable & \FPT wrt.\ $t$ and \XP wrt.\ $\ell$ &\autoref{thm:partitioning-properties}\\
    Staggered & \W{1}-hard wrt.\ $(k, \ell)$ &\autoref{thm:meta}\\
    \midrule
    \multirow{2}{*}{\gprop{Matching}} & polynomial-time solvable for all $\ell\le 2$& \multirow{2}{*}{\autoref{thm:matching}}\\
    & \W{1}-hard wrt.\ $k$ for all $\ell\ge3$&\\
    \gprop{$c$-Factor} & \W{1}-hard wrt.\ $k$ for all $\ell\ge2$ &
    \autoref{thm:cfactor}\\
        \gprop{Hamiltonian} & \W{1}-hard wrt.\ $k$ for all $\ell\ge2$ & \autoref{thm:hamiltonian}\\
    \bottomrule
    \end{tabular}
\end{table}

\paragraph{Our Results.} We give an overview of our results in \autoref{tab:results}. Observe that for \PML, \NP-hardness and \W{1}-hardness for either $k$ or $|V|-k$ in the single-layer case directly imply hardness of the multi-layer case.
 Our analysis of \PML starts with several easy observations on \emph{hereditary}
graph properties~$\Pi$, that is, $\Pi$ is closed under vertex deletions. Such properties~$\Pi$ have been well-studied in the
single-layer case. Using Ramsey arguments and a theorem due to \citet{KR02}, we 
show a trichotomy for the complexity of \PML with respect to the inclusion of edgeless or complete graphs in~$\Pi$, distinguishing between polynomial-time solvability, fixed-parameter
tractability with respect to the combined parameter~$(k, \ell)$, and
\W{1}-hardness with respect to~$k$ for all~$\ell$ (\autoref{thm:hereditary}). Second, we generalize
a result due to \citet{cai1996fixed} by showing that, for graph properties~$\Pi$
characterized by a finite number of forbidden induced subgraphs, \PML
is fixed-parameter tractable with respect to the combined parameter~$(t - \ell, |V|- k)$, and that in this case it additionally admits a polynomial-size problem kernel
(\autoref{prop:forb-subgr}).

Next, we turn to graph properties that are not necessarily
hereditary. An easy example would be connectedness. For finding connected graphs
of order at least~$k$ in at least~$\ell$ of~$t$~layers, there is a simple \FPT{}-algorithm
with respect to~$t$ which is also an \XP-algorithm with respect
to~$t - \ell$ or with respect to~$\ell$. This algorithm admits a
generalization to each graph property that implies certain
good-natured partitions of the input graphs
(\autoref{thm:partitioning-properties}), for example so-called $c$-cores~\cite{Sei83} and
$c$-trusses~\cite{Coh08}. On the flip side, we spot \W{1}-hardness for
\PML\ for the combined parameter~$k$ and~$\ell$ for a specific large class of graph properties~$\Pi$ which we call \emph{staggered} (\autoref{thm:meta}). Such properties allow to efficiently construct graphs which have three parts which are obligatory, optional, and forbidden for the solution, respectively. Staggered graph properties~$\Pi$ include connected graphs, $c$-cores, and
$c$-trusses, for example (see \autoref{cor:hardness}).

Finally, we exhibit three graph properties~$\Pi$ for which \autoref{thm:meta} and \autoref{cor:hardness} already yield hardness, but where we can achieve stronger results by closer inspection. 
First, we consider the property that includes all graphs that admit a perfect matching. While finding a vertex subset of size~$k$ that induces subgraphs with a perfect matching in two layers is polynomial-time
solvable, it becomes \NP-hard and \W{1}-hard with respect to the number of vertices to select~$k$ in
three layers (\autoref{thm:matching}). For a generalization of matchings,
so-called $c$-factors, the subgraph detection problem already becomes \NP-hard
and \W{1}-hard with respect to the number of vertices to select~$k$ in two layers for all~$c\ge 2$~(\autoref{thm:cfactor}). 
Furthermore, we consider the property that includes all graphs that admit a Hamiltonian path.
While finding an order-$k$ subgraph containing a Hamiltonian path is fixed-parameter tractable with
respect to the number~$k$ of vertices to select in one layer~\cite{monien1985find}, it becomes
\W{1}-hard in two layers (\autoref{thm:hamiltonian}). 

Apart from providing a broad overview over the complexity of \PML, the main technical contributions are revealing conditions on $\Pi$ that make \PML\ hard
(\autoref{thm:meta}) 
and understanding the transition from
tractability to hardness for perfectly matchable subgraphs
(\autoref{thm:matching}) and Hamiltonian subgraphs
(\autoref{thm:hamiltonian}). 

\paragraph{Related Work.} \looseness=-1 As mentioned in the beginning,
despite the numerous practical studies related to multi-layer
networks, systematic work pertaining to the computational complexity
of \PML\ is not well-developed. The following special cases were
studied from this viewpoint. \citet{GHPR03} and
\citet{BHP08} studied the case where the
graph property~$\Pi$ is the set of connected graphs and
$t = \ell = 2$, that is, they studied the problem of finding connected
subgraphs of size at least~$k$ in two layers. They showed that the resulting problem is
polynomial-time solvable. In contrast, \citet{CY14}
studied a modified version of this problem, where the desired vertex
subset shall be of size \emph{exactly}~$k$ instead of at
least~$k$.\footnote{Note that for hereditary graph properties requiring
  vertex sets of size exactly~$k$ or at least~$k$ is equivalent in
  terms of computational complexity. However, for connectivity, which
  is not a hereditary property, it can make a difference.} 
  They showed
\NP-hardness and \W{1}-hardness with respect to the number~$k$ of
vertices to select and with respect to~$|V| - k$. \citet{ALMS15} gave a~$23^{tk} \cdot \poly(n, t)$-time
algorithm for the case where~$\Pi$ is the set of acyclic graphs and
$t = \ell$.

In terms of general graph properties, \citet{CY14}
proved a trichotomy for hereditary graph properties similar to the one
we give in \autoref{thm:hereditary} in the setting where the input
consists only of two layers whose edge-sets are disjoint and one wants
to satisfy two possibly distinct graph properties in the corresponding
layers (see their Theorem~6). They also showed that \PML\ is
fixed-parameter tractable parameterized by~$|V| - k$ in the following
modified setting: Each layer~$i$ has a specific graph property~$\Pi_i$
which is characterized by a finite set of forbidden induced subgraphs
and the vertex sets of the layers may differ (see their Theorem~7).
Our \autoref{prop:forb-subgr} is strongly related to this result and
when focusing on the plain \PML\ problem, it can be seen as a
generalization of their result and additionally provides a
polynomial-size problem kernel. 

In mathematical terms, multi-layer graphs are equivalent to edge-colored multigraphs. These have been studied from an algorithmic viewpoint; an overview can be found in the surveys of Bang-\citet[Chapter~26]{BG09}, and \citet{KL08}. Most of the algorithmic results presented there pertain to paths and cycles which do
not contain two consecutive edges in the same layer and to related
questions like connectedness and Hamiltonicity using this notion of
paths or cycles.

If a multi-layer graph is additionally equipped with a linear ordering of the layers, then the model is mathematically equivalent to \emph{temporal
  graphs}~\cite{holme2012temporal,holme2015,michail2016introduction,latapy2017stream}. In this model, each layer models the state of the data set at a different point in time. There is a huge body of research in this area and we refer to the aforementioned surveys for an overview. 

\paragraph*{Organization.} 
\Cref{sec:prelim} contains some basic definitions and notation. In \Cref{sec:hereditary} we give general results for hereditary graph
properties~$\Pi$. In \Cref{sec:meta} we give general results for large classes of graph
properties~$\Pi$ that are not necessarily hereditary. In the two succeeding sections, we take a closer look at selected graph properties where a more in-depth inspection reveals stronger results compared to the general result of~\Cref{sec:meta}. In \cref{sec:mlmatching} we investigate the parameterized computational complexity of \gprop{Matching}-\ML{} and the generalization \gprop{$c$-Factor}-\ML{}. 
In \cref{sec:path-hard} we analyze the parameterized computational complexity of \gprop{Hamiltonian}-\ML{}. 
We give a conclusion and directions for future work in \Cref{sec:conclusion}.

\section{Preliminaries}\label{sec:prelim}

\paragraph*{Parameterized Complexity.}

A \emph{parameterized problem} is a language
$\paramProblem \subseteq \Sigma^* \times \mathbb{N}$, where the second component
in an instance $(I, k) \in \Sigma^* \times \mathbb{N}$ is called the
\emph{parameter}. In the case of combined parameters, we write
  a tuple, e.g.\ $(k_1, k_2)$. This is just notation for a parameter
  $k = k_1 + k_2$.
  A parameterized
problem~$\paramProblem$ is \emph{fixed-parameter tractable} if there is an
algorithm that for each instance
$(I, k) \in \Sigma^* \times \mathbb{N}$ decides whether $(I, k) \in \paramProblem$
in~$f(k) \cdot |I|^{O(1)}$ time, where~$f$ is a computable function
and $|I|$ is the encoding length of the input~$I$. We also say that $|I|$ is the
\emph{instance size}. The class of fixed-parameter tractable problems
is~\FPT. A parameterized problem $\paramProblem$ is in the class \XP\ if there is
an algorithm that decides for each instance
$(I, k) \in \Sigma^* \times \mathbb{N}$ whether $(I, k) \in \paramProblem$ in~$|I|^{f(k)}$ time, where~$f$ is a computable function.
\W{1}-hard parameterized problems are generally assumed not to be
fixed-parameter tractable. \W{1}-hardness can be shown by a
parameterized reduction from another \W{1}-hard problem such as
\textsc{Independent Set} (given a graph~$G$ and an integer $k$, decide
whether there is a $k$-vertex subset in~$G$ that does not contain any
edge). A \emph{parameterized reduction} from a parameterized
problem~$Q$ to a parameterized problem~$\paramProblem$ is an algorithm that maps
an instance~\((I,k)\) of~\(Q\) to an instance~$(I',k')$ of~\(Q\) in
$f(k) \cdot |I|^{O(1)}$~time such that $(I,k) \in Q$ if and only
$(I',k') \in \paramProblem$ and $k'\leq g(k)$, where~\(f\)~and~\(g\) are arbitrary
computable functions.

Furthermore, parameters allow us to mathematically rigorously study
efficient data reduction. Formally, given a parameterized problem~$\paramProblem$, a \emph{kernelization algorithm} is a poly\-nomial-time algorithm that maps instances~$(I, k)$ of~$\paramProblem$ to instances~$(I', k')$ (called a \emph{problem kernel}) of~$\paramProblem$ such that the size of $I'$ is
upper-bounded by a function of the parameter~$p$ and $(I', k')$ is a yes-instance if and only if $(I, k)$ is a yes-instance. If the kernel size can be upper-bounded by a polynomial in the parameter, we call it a \emph{polynomial kernel}.
 For more context and methodology we refer to the
literature~\cite{Cy15,DF13,FG06,Nie06}.


\paragraph*{Graphs.}  
All graphs in this work are undirected and without self loops or parallel edges. We
use standard graph notation~\cite{Die05}. A graph property~$\Pi$ is \emph{hereditary} if removing any vertex from a graph in~$\Pi$ results again in a graph in~$\Pi$. We consider the following graph properties. A graph is a \emph{$c$-core} if each vertex has degree at least~$c$~\cite{Sei83}. A graph is a \emph{$c$-truss} if it is connected and each edge is contained in at least~$c-2$ triangles~\cite{Coh08}. We say that a graph is \emph{Hamiltonian} if it contains a simple path that comprises all vertices in the graph. The \emph{length of a path} is the number of its edges. A \emph{$c$-factor} in a graph is a subset of the edges such that each vertex is incident with exactly $c$~edges. In sans serif font face we often denote graph properties. For example, \gprop{$c$-Truss} is the set of all $c$-trusses.
By \gprop{Matching} we refer to the set of all graphs containing a perfect matching; note that then finding a maximum matching is equivalent to finding a maximum sized \gprop{Matching}-subgraph. By \gprop{$c$-Factor} we refer to the set of all graphs containing a $c$-factor. For a list of definitions of all graph properties mentioned in this article, see Appendix~\ref{appendix:definitions}.

\section{Hereditary Graph Properties}\label{sec:hereditary}
In this section we study the (parameterized) computational complexity of \PML with respect
to hereditary graph properties~$\Pi$. Many natural graph properties fall into this
category, for example being planar or being a forest. We give a trichotomy of the complexity,
classifying each problem either as polynomial-time solvable or as \NP-hard, and further
classifying the parameterized complexity of the \NP-hard cases with respect to the parameters number~$k$ of vertices to select and number~$\ell$ of layers to select.
In addition, we observe fixed-parameter tractability for the 
deletion parameters~$|V| - k$ and~$t - \ell$.

The single-layer case has been studied by \citet{lewis1980node} as well as \citet{KR02};
the latter studied the parameterized complexity of the subgraph detection problem for hereditary properties. \citet{CY14} studied this problem on two layers with disjoint edge sets (see their Theorem~6) and for multiple layers when layers cannot be deleted (in their Theorem~7). We generalize below
the mentioned results to the multi-layer case where layers can be deleted. This allows us to classify all hereditary graph properties $\Pi$ by the parameterized complexity of the corresponding \PML problem.
\begin{prop}[Complete classification of hereditary graph properties]
\label{thm:hereditary}
If $\Pi$ is a hereditary graph property, then the following statements are true.
\begin{compactenum}
  \item If $\Pi$ excludes at least one complete graph and at least one edgeless
  graph, then \PML is solvable in polynomial time.
  \item If $\Pi$ includes all complete graphs and all edgeless graphs, then
  \PML is \NP-hard and \FPT when parameterized by the combined parameter number~$k$ of vertices to select and number~$\ell$ of layers to select.
\item If $\Pi$ includes either all complete graphs or all edgeless graphs (but not both), then
\PML is \NP-hard and \W{1}-hard when parameterized by the number~$k$ of vertices to select for all numbers~$\ell$ of layers to select.
\end{compactenum}
\end{prop}%
\begin{proof}
We utilize the concept of \emph{Ramsey numbers}.
The Ramsey number $R(p, q)$ is the minimum number~$x$ such that every graph with~$x$ vertices has either a clique of size $p$ or an
independent set of size $q$. It is well-known that $R(p, q) \leq
{p+q-2 \choose q-1}$~\cite{Jukna11}. We give separate proofs for
all three statements in the theorem. Note that for the second and third statement,
\NP-hardness even in the single-layer case was shown by \citet{lewis1980node}.

\proofparagraph{Statement 1.} Let $p, q$ be the sizes of the smallest
excluded complete and edgeless graph, respectively. Note that any graph on at
least $R(p, q)$ vertices contains either a clique of size~$p$ or an independent
set of size $q$ and hence is not included in $\Pi$. Hence, there are only
finitely many graphs that have property $\Pi$. Furthermore, if $k
\geq R(p, q)$, then we face a no-instance.

If $k < R(p, q)$, then we consider every order-$k$ vertex subset $X$ and check whether $G[X] \in \Pi$ in at least $\ell$ layers. If this is the case for some~$X$, then we output~$X$; otherwise, we output that the instance is a no-instance. The running time for this algorithm is~$O(t{n \choose k}f(k))$ where $f(k)$ is the time to check
membership of $\Pi$. Note that $k$ is constant since~$k<R(p,q)$ and $p$ and $q$ are
constants only depending on $\Pi$. Hence, the overall running time is polynomial.
 
\proofparagraph{Statement 2.} The \NP-hardness follows from the \NP-hardness of the single-layer case. For the proof of fixed-parameter tractability, we introduce nested Ramsey
numbers as follows.
\begin{align*}
R^{(1)}(p, q) &= R(p, q),\\
R^{(i)}(p, q) &= R(R^{(i-1)}(p, q), R^{(i-1)}(p, q)).
\end{align*}
We show that if $|V| \geq R^{(\ell)}(k,k)$, then we face a yes-instance. Indeed, we show the more general statement that, for each set of~$\ell$ layers~$G_i$, $i = 1, \ldots, \ell$, on vertex set $V$ with $|V| \geq R^{(\ell)}(k,k)$, there is a vertex subset $X \subseteq V$ with $X \geq k$ such that $G_i[X] \in \Pi$ for each~$i \in 1, \ldots, \ell$. We prove this by induction on $\ell$.

For $\ell=1$ we have $|V| \geq R(k, k)$. Hence, each graph on
vertex set~$V$ contains either a clique of size $k$ or an independent
set of size $k$, proving the statement. Assume that $\ell > 1$ and that
the statement holds for each $\ell' < \ell$. Since
$|V| \geq R^{(\ell)}(k, k)$, each layer has either a clique of size
$R^{(\ell-1)}(k, k)$ or an independent set of size
$R^{(\ell-1)}(k, k)$. Let $X' \subseteq V$
with $|X'| = R^{(\ell-1)}(k, k)$ be either a clique or an independent
set in layer~1. By the induction hypothesis, there is a vertex
set $X\subseteq X'$ with~$|X|\geq k$ such that $G_i[X] \in \Pi$ for
all~$i$ such that~$2\le i\le \ell$. Since~$X\subseteq X'$, we also have~$G_1[X] \in \Pi$. Hence, $G_i[X] \in \Pi$ on each of the
$\ell$ layers~$G_i$, as required.

The algorithm is now as follows. If $|V| \geq R^{(\ell)}(k, k)$, then
accept immediately. By the above, each subset of $\ell$ layers of the
input multi-layer graph contains a solution. Otherwise, if
$|V| < R^{(\ell)}(k, k)$, then find a solution by brute force, if it
exists: Simply try all possible vertex subsets~$X$ of size~$k$ and
check whether $G_i[X] \in \Pi$ for at least~$\ell$
layers~$i$. By heredity, if there is a solution, then there is one of size~$k$ and thus the algorithm is correct. If $g(k)$ denotes the time
needed to check whether $G_i[X] \in \Pi$ for some fixed $i$, then the
running time for the brute-force step is at most
$|V|^{k + O(1)}\cdot t \cdot g(k) \leq (R^{(\ell)}(k, k))^{k + O(1) } \cdot g(k) \cdot t = f(k, \ell) \cdot t$ for some function~$f$, showing that the problem is \FPT\ with respect to $k$ and $\ell$ combined.



\proofparagraph{Statement 3.} \citet{KR02} showed that for
hereditary properties $\Pi$ including either all complete graphs or all edgeless graphs,
$\Pi$-\SL is \W{1}-hard when parameterized by $k$. This directly translates to the multi-layer case, as does \NP-hardness.
\end{proof}
Note that every hereditary graph property falls into one of the three cases of \autoref{thm:hereditary}.
Properties that fall into the first case are exactly those containing
only a finite number of graphs. In the following corollary, we give a number of hereditary properties $\Pi$
that fall in the second and third case and give the corresponding complexity results for \PML implied by
\autoref{thm:hereditary}. For their definitions we refer to the
literature~\cite{BLS99, Gol04} or to Appendix~\ref{appendix:definitions}.

\begin{cor}
\label{cor:fpt}
\begin{compactenum}
\item \PML is \NP-hard and \FPT when parameterized by the combined parameter number~$k$ of vertices to select and number~$\ell$ of layers to select for
$\Pi \in \{$\gprop{Asteroidal Triple Free
Graph}, \gprop{Chordal Graph}, \gprop{Comparability Graph}, \gprop{Interval Graph}, \gprop{Perfect Graph}, \gprop{Permutation Graph}, \gprop{Split Graph}$\}$.

\item \PML is \NP-hard and \W{1}-hard when parameterized by the number~$k$ of vertices to select for all numbers~$\ell$ of layers to select for
$\Pi \in \{$\gprop{$c$-Colorable Graph}, \gprop{Complete Graph}, \gprop{Complete Multipartite Graph}, \gprop{Edgeless Graph}, \gprop{Forest}, \gprop{Planar Graph}$\}$.
\end{compactenum}
\end{cor}
Next, we consider properties $\Pi$ whose complements are hereditary or, equivalently, $\Pi$ is closed under the operation of adding a new vertex~$v$ and connecting $v$ arbitrarily to the rest of the graph. For these we
can observe that polynomial-time solvability transfers to the
multi-layer case.
\begin{obs}
\label{thm:easyness}
Let $\Pi$ be a graph property such that whenever $G \in \Pi$ for some graph~$G$, then
we have that for all graphs $H=(V, E)$ if there is a vertex set $X\subseteq V$ such that $H[X]$ is isomorphic to~$G$, then $H \in \Pi$. 
If $\Pi$ can be decided in
$T(|V|)$ time for some function $T$, then \PML can be decided in $O(t\cdot T(|V|))$ time for
all numbers~$k$ of vertices to select and all numbers~$\ell$ of layers to select.
\end{obs}
\begin{proof}
Let $\Pi$ be a graph property such that if $G \in \Pi$ for some graph $G$ and
$H[X] = G$ for some graph $H$ and vertex set $X$, then $H \in \Pi$. Observe that if $G \notin \Pi$, then no induced subgraph of~$G$ can be in $\Pi$.
Let $(G_1, \ldots, G_t, k, \ell)$ be an instance of \PML. We decide for each graph $G_1, \ldots, G_t$ whether it satisfies property~$\Pi$. We
face a yes-instance if and only if there are at least $\ell$ graphs that have
property~$\Pi$: We can set $X = V$, and hence $|X| \geq k$, for any $k \leq n$.
\end{proof}
In the following corollary, we give two examples of properties $\Pi$ for which by \autoref{thm:easyness}
\PML is solvable in polynomial time.
\begin{cor}
\label{cor:easyness}
\PML is solvable in polynomial time for:
\begin{compactitem}
  \item $\Pi =$ ``The graph has maximum degree at least $x$.''
  \item $\Pi =$ ``The graph contains a triangle.''
  \item $\Pi =$ ``The graph has an $h$-index\footnote{The \emph{$h$-index} of a graph is the largest integer $h$ such that the graph contains at least $h$ vertices with degree at least $h$~\cite{ES12}.} of at least $x$.''
\end{compactitem}
\end{cor}

Finally, we consider the dual parameterizations for hereditary graph
properties characterized by a finite number of forbidden subgraphs. In
the single-layer case, this problem has been studied by \citet{lewis1980node} as well as
\citet{cai1996fixed}. A result of \citet[Theorem~7]{CY14} implies a fixed-parameter algorithm for \PML\ if
$t - \ell = 0$, that is, no layers can be deleted. Below we give a
more general fixed-parameter algorithm for the case where
$t - \ell \geq 0$ and we furthermore present a polynomial kernel.
\begin{prop}\label{prop:forb-subgr}
  Let $\Pi$ be a hereditary graph property that is characterized by finitely
  many forbidden induced subgraphs. Then \PML{} is \NP-hard and \FPT when
  parameterized by the combined parameter number~$t - \ell$ of layers to delete and number~$|V|
  - k$ of vertices to delete. It also admits a polynomial-size problem
  kernel with respect to this combined parameter.
\end{prop}
  \begin{proof}
  To see the fixed-parameter tractability, consider the search-tree algorithm that recursively searches for a forbidden induced subgraph~$G'$ in one of the layers, and branches, for each vertex~$v$ in~$G'$, into the branch of deleting~$v$ and, additionally, into the branch of deleting the layer of~$G'$. Hence, in each branch we either delete a layer or a vertex, so the depth of the search-tree is upper-bounded by~$(t - \ell) + (|V| - k)$. Finding~$G'$ takes polynomial time because there is only a constant number of different forbidden subgraphs and each one has constant size. Furthermore, each node in the resulting search tree has a constant number of children. Hence, the search-tree algorithm has a running time of $c^{t - \ell + |V| - k} \cdot \poly(I)$, where $c$ is a constant, and $I$ is the instance size, as required. This procedure is correct since each forbidden subgraph is either destroyed by deleting one of its vertices or by deleting the layer it exists in. Hence each of the remaining layers does not contain a forbidden subgraph and therefore has property~$\Pi$.

  \newcommand{\HS}{\textsc{Hitting Set}}
  \newcommand{\TCHS}{\textsc{2-Color Hitting Set}}
  To see that \PML{} admits a polynomial kernel with respect to the combined parameter $(t - \ell, |V| - k)$, we use a reduction 
  to \TCHS{}, a variant of \HS, and then apply a (basically folklore) kernelization for \TCHS{}. Herein, we are given two disjoint ground sets~$B$ and~$W$, a family~$\F$ of subsets of~$B \cup W$, and two integers~$b, w$. We are to decide whether there is a \emph{hitting set} $S \subseteq B \cup W$, that is, each subset~$F \in \F$ has $F \cap S \neq \emptyset$, containing $b$~elements from~$B$ and $w$~elements from~$W$. Clearly, \TCHS{} is contained in \NP.

  The reduction works as follows. Given an instance of \PML{}, we put the ground set~$B := V$ and put a distinct new vertex into $W$ for each layer. For each layer, we enumerate all forbidden induced subgraphs. This takes polynomial time, as the maximum size of these subgraphs is a constant. To define~$\F$, for each forbidden induced subgraph~$G'$ we add its vertex set~$V'$ plus the vertex~$v \in W$ corresponding to the layer in which~$G'$ is contained as a set $V' \cup \{v\}$ to~$\F$. Integer~$b$ is set to~$|V| - k$ and integer~$w$ to $t - \ell$. As mentioned, the reduction works in polynomial time. Since we have to ``hit'' each forbidden induced subgraph by either deleting a vertex from it, or deleting its layer completely, it is not hard to verify that the reduction is correct.

  We now apply the so-called sunflower kernelization procedure~\cite{Mos09PHD,Kra09,Bev14} to the resulting \TCHS{} instance. A \emph{sunflower} in~$\F$ is a subfamily~$\F' \subseteq \F$ such that there is a set $C \subseteq B \cup W$ with the property that each pair $F, F' \in \F'$ has $F \cap F' = C$. The \emph{size} of a sunflower is~$|\F'|$. If there is a sunflower of size~$b + w + 1$ in $\F$, then every hitting set contains at least one element of~$C$. Hence, we can safely remove one set out of every sunflower of size at least~$b + w + 2$. This can be done exhaustively in polynomial time~\cite{Mos09PHD,Kra09,Bev14}. After this procedure has been carried out, Erd\H{o}s and Rado's 
  Sunflower Lemma~\cite{ER60}  guarantees that the remaining set family~$\F$ has size $O((b + w)^c)$, where $c$ is the size of the largest set in~$\F$. This is a polynomial because the sets in $\F$ have constant size. By removing elements of $B \cup W$ which are not contained in any set in~$\F$, we obtain an overall size bound on the resulting instance of \TCHS{} which is polynomial in $b = |V| - k$ and $w = t - \ell$. 

  Finally, we transfer the instance of \TCHS{} created in this way to an
  equivalent instance of \PML{} by using a polynomial-time many-one
  reduction. Such a reduction exists because \TCHS{} is in \NP{} and \PML{}
  is \NP-hard for every graph property $\Pi$ that is characterizable by a finite
  number of forbidden induced subgraphs~\cite{lewis1980node}.
\end{proof}
In the following corollary, we give several hereditary properties $\Pi$ that are
characterized by a finite number of forbidden subgraphs. For these,
\PML is fixed-parameter tractable with respect to the
  number~$t - \ell$ of layers to delete and the number~$|V| - k$ of vertices to
  delete combined. For their definitions see Appendix~\ref{appendix:definitions}.
  \begin{cor}
    Let $\Pi \in \{$\gprop{Cluster Graph}, \gprop{Cograph},
    \gprop{Line Graph}, \gprop{Quasi-Threshold Graph}, \gprop{Split Graph}$\}$.  
    Then, \PML is \NP-hard. When parameterized by the combined parameter number~$t - \ell$ of layers delete and number~$|V| - k$ of vertices to delete, then \PML\ is \FPT\  
and admits a polynomial kernel.
\end{cor}
While this collection of graph properties might seem motivated only from a graph-theoretic point of view, we note that quasi-threshold graphs are a model for familial groups in social networks~\cite{NG13} and split graphs are a model of communities with core/periphery structures~\cite{BE00}.

\section{Non-hereditary Graph Properties}\label{sec:meta}

In contrast to the previous section, we now give two results related to graph
properties that are not necessarily hereditary. Two examples of such non-hereditary properties that are prominent in network analysis are \gprop{Connectivity} and  \gprop{$c$-Core} (recall that graph is a $c$-core if each vertex has degree at least~$c$). First, we consider graph properties in which each graph admits a certain
nice vertex partitioning. This family of graph properties includes \gprop{Connectivity} and \gprop{$c$-Core} among others. We give an \FPT-algorithm for the parameter total
number~$t$ of layers; this algorithm is also an \XP-algorithm with respect
to the number~$\ell$ of subgraph layers. Second, we provide a
general sufficient condition for graph properties~$\Pi$ for which \PML is
\W{1}-hard for the combined parameter~$(k, \ell)$.
This general sufficient condition 
captures
many prominent graph properties such as \gprop{Connectivity}, \gprop{$c$-Core}, \gprop{$c$-Truss}, and \gprop{Matching}.

\paragraph{Vertex-partitionable graph properties.} We start with investigating graph
properties~$\Pi$ that allow for efficiently computable partitions of
a graph into maximal components that each satisfy~$\Pi$.  For such properties~$\Pi$ it turns
out that finding large $\Pi$-subgraphs in all layers of the input graph is
tractable. This can be seen as a generalization of the component-detection
algorithm in two layers by \citet{GHPR03}.

First, we define the type of partition that the graph property shall allow.
\begin{defi}\label{def:partitioning-properties}
  Let~$\Pi$ be a graph property and let~$G=(V,E)$ be a graph. A partition~${\cal P}:=\{X_1, \ldots , X_x\}$ of~$V$ is a~\emph{$\Pi$-partition} if
  \begin{compactenum}[(i)]
  \item $G[X_i]\in \Pi$ for all~$X_i\in {\cal P}$,\label{enum:pp:issol}
  \item for all~$X\subseteq V$ such that~$G[X]\in \Pi$, we
    have~$X\subseteq X_i$ for some~$X_i\in {\cal P}$.\label{enum:pp:contsol}
  \end{compactenum}
\end{defi}
Informally, the existence of a~$\Pi$-partition means that there are maximal components containing all induced subgraphs that fulfill the property~$\Pi$. In the case of~\gprop{Connectivity}, these are exactly the connected components of the graph. We now show that if~$\Pi$-partitions can be computed efficiently, then we can solve \PML efficiently if~$t$,~$\ell$, or~$t-\ell$ are small.
\begin{prop}
\label{thm:partitioning-properties}
  Let~$\Pi$ be a graph property such that every graph~$G=(V,E)$ has a~$\Pi$-partition that can be computed in~$T(|V|,|E|)$ time
    where~$T$ is non-decreasing in both arguments.
  Then, \PML can be solved in~$\binom{t}{\ell} \cdot
  O(|V|\cdot t)\cdot \max_{1\le i\le t} (|E_i| + T(|V|,|E_i|))$ time.
\end{prop}
\begin{proof}
  We describe an algorithm that works when $\ell=t$, that is, when we aim
  for satisfying property~$\Pi$ in all layers.
  To apply this algorithm for all other cases, we iterate through all possibilities to select $\ell$~layers where property~$\Pi$ shall be satisfied
  and then apply the algorithm for this selection.
  This gives an additional factor of $\binom{t}{\ell}$ to the overall running time.
  In each application, our algorithm outputs all maximal sets~$X\subseteq V$
  such that~$G_i[X]\in \Pi$ for all selected input graphs~$G_i$.
  We refer to these sets as \emph{solutions} in the following.
  The algorithm maintains a partition~${\cal P}$ of~$V$ where, initially,~${\cal P}=\{V\}$.
  
  The algorithm checks whether there is a~$Y\in {\cal P}$ such
  that~$G_i[Y]\notin \Pi$ for some input graph~$G_i$.
  If there is such a vertex set~$Y\in {\cal P}$, then it computes
  in~$T(|V|,|E_i|)$ time a~$\Pi$-partition ${\cal P}_Y$ 
  of~$G_i[Y]$. The partition~${\cal P}$ is replaced by~$({\cal
    P}\setminus \{Y\}) \cup {\cal P}_Y$. If there is no such vertex set, then
  the algorithm outputs all~$Y\in {\cal P}$. It accepts if one of the outputs has size at least~$k$ and rejects otherwise.
  
  To see the correctness of the algorithm, first observe that, for
  each output~$Y$, we have that~$G_i[Y]\in \Pi$
  for all selected input graphs~$G_i$. To show maximality of each~$Y$, we show
  that the algorithm maintains the invariant that each solution~$X$ is
  a subset of some~$Y\in {\cal P}$. This invariant is trivially
  fulfilled for the initial partition~$\{V\}$. Now consider a set~$Y$
  that is further partitioned by the algorithm. By the invariant, any
  solution~$X$ that has nonempty intersection with~$Y$ is a subset
  of~$Y$. Since~${\cal P}_Y$ is a~$\Pi$-partition
  of~$G_i[Y]$, by Property~(\ref{enum:pp:contsol}) of $\Pi$-partitions, there is no solution~$X$ that contains vertices of two
  distinct sets~$Y_1$,~$Y_2$ of~${\cal P}_Y$. Thus, each solution that
  is a subset of~$Y$ is also a subset of some~$Y'\in {\cal
    P}_Y$. Hence, each output set~$X$ is a solution as it is an
  element of the final partition~${\cal P}$ and all solutions are
  subsets of elements of~${\cal P}$.

  To upper-bound the running time, observe that for each~$Y\in {\cal P}$, we
  can test in~$O(t\cdot \max_{1\le i\le t} T(|V|,|E_i|))$ time
  whether it needs to be partitioned further. At most~$|V|$ partitioning
  steps are performed and if a set~$Y\in {\cal P}$ does not need to be
  partitioned further, then it can be discarded for the remainder of
  the algorithm. Thus, in~$O(|V|)$ applications of the ``maximality
  test'' the result is that~$Y$ is a solution and in~$O(|V|)$
  applications of the maximality test,~$Y$ is further
  partitioned. Hence, the overall number of sets~$Y$ that are elements
  of~${\cal P}$ at some point is~$O(|V|)$. The overall running time
  now follows from the assumptions on~$T$ and from the fact that
  the induced subgraphs for all~$G_i$ can be computed in~$O(t\cdot
  \max_{1\le i\le t}|E_i|)$ time for each~$Y$.
\end{proof}
Examples of graph properties covered by \autoref{thm:partitioning-properties}
are \gprop{Connectivity}, and \gprop{$c$-Edge-Connectivity}.
If we assume that graphs on one vertex are considered as
(trivial) $c$-cores, then the \gprop{$c$-Core} property is also covered: the
nontrivial~$c$-core of a graph is uniquely determined (it is the
subgraph remaining after deleting any vertex with degree less
than~$c$). Similarly, the~\gprop{$c$-Truss} property is covered by \autoref{thm:partitioning-properties} if we
allow one-vertex graphs to be considered as~$c$-trusses. Observe that
we can also require the~$c$-cores and~$c$-trusses to be connected. For definitions of the graph properties mentioned above and in the following corollary, see Appendix~\ref{appendix:definitions}.

If~$T$ is a polynomial function, which holds for all examples described above,
then~\PML is fixed-parameter tractable with respect to~$t$ and polynomial
time solvable if~$\ell$ or~$t-\ell$ are constants.

\begin{cor}
  \label{cor:partitioning-fpt}
  Let $\Pi \in
 \{$\gprop{$c$-Core}, \gprop{$c$-Edge-Connectivity}, \gprop{Connectivity}, \gprop{$c$-Truss}$\}$.
 \PML is \FPT when parameterized by the total number~$t$ of layers and
 polynomial-time solvable if the number~$\ell$ of layers to select or
 the number~$t-\ell$ of layers to delete are constants.
\end{cor}

\paragraph{A general hardness reduction.} We now describe a large class of graph properties~$\Pi$ for
which \PML is \NP-hard and \W{1}-hard when parameterized by the combined parameter number~$k$ of vertices to select and number~$\ell$ of layers to select.
We call these graph properties \emph{staggered} (\autoref{def:staggered}). The definition is somewhat technical but covers many natural
graph properties~$\Pi$ which are not hereditary, such as \gprop{Connectivity}. 
In the main theorem of this paragraph (\autoref{thm:meta}) we show that for staggered graph properties $\Pi$, \PML is \NP-hard and \W{1}-hard when parameterized by the combined parameter number~$k$ of vertices to select and number~$\ell$ of layers to select.
Furthermore, since all graph properties from
\autoref{cor:partitioning-fpt} are staggered, it shows that for those properties \PML 
becomes intractable when parameterized by~$\ell$ instead of~$t$. We list some graph properties~$\Pi$ for which \autoref{thm:meta} can be applied in
\autoref{cor:hardness} and explain how the theorem yields the results in said cases.

Intuitively, \autoref{thm:meta} covers graph properties~$\Pi$ for which it is possible to construct graphs~$G$ with the following properties:
\begin{compactitem}
\item The vertices are partitionable in three sets: a fixed number of \emph{obligatory} vertices, a variable number of \emph{optional} vertices, and a variable number of \emph{forbidden} vertices.
\item A sufficiently large subgraph of~$G$ has property~$\Pi$ if and only if it contains all obligatory vertices and no forbidden vertices.  
\end{compactitem}
Formally, we characterize these graph properties~$\Pi$ as follows.
\begin{defi}\label{def:staggered}
Let $\Pi$ be a graph property and~$f, g: \mathds{N}\to\mathds{N}$ two polynomial-time computable functions. We say that~$\Pi$ is $(f,g)$-\emph{staggered} if there is an
algorithm $A_{\Pi}$ and a constant~$c_{\Pi}$ depending only on $\Pi$ such that~$A_{\Pi}$ takes as input a set $W$, a subset $W'
\subseteq W$, and an integer $\alpha \ge c_{\Pi}$, runs in $(|W|+\alpha)^{O(1)}$ time, and computes a graph $G=(V, E)$ 
fulfilling the following conditions.
\begin{compactenum}
  \item For each $w \in W$ there is a vertex set $X_w \subseteq V$ with $|X_w| =
  f(\alpha)$,
  \item $\{X_w\mid w\in W\}\cup \{Y\}$ is a partition of~$V$ for some $Y$ with $|Y| = g(\alpha)$,
  \item for all $X \subseteq V$ with $|X|\geq \alpha\cdot
  f(\alpha) + g(\alpha)$ we have that
  \[ 
  G[X] \in \Pi \ \Leftrightarrow\  \exists W'' \subseteq W'\text{ such that } X = \bigcup_{w\in W''} X_w \cup Y.
  \]
\end{compactenum}
The output of Algorithm~$A_{\Pi}$ is the graph $G=(V, E)$ as well as the partition of its vertex set $\{X_w\mid w\in W\}\cup \{Y\}$.
\end{defi}
The intuition is that each set $X_w$ corresponds to one vertex $w\in
W$ and every sufficiently large set $X$ such that~$G[X]\in \Pi$ either fully
contains~$X_w$ or not. Furthermore, $Y$ is the set of obligatory vertices that have to
be included in $X$ in order to have that $G[X]\in \Pi$. Finally, the sets~$X_w$ that correspond to vertices in~$w\in W\setminus W'$ are forbidden, that is, they have to be
fully excluded from $X$ in order to have that $G[X]\in \Pi$.  For the
proof, we reduce from \textsc{Biclique}, which is~\NP-hard~\cite{John87a} and~\W{1}-hard when parameterized by the size~$h$ of the
biclique~\cite{lin2015parameterized}.
\begin{theorem}
\label{thm:meta}
Let $\Pi$ be an $(f,g)$-staggered graph property. Then,~\PML is \NP-hard and \W{1}-hard
when parameterized by the combined parameter number~$k$ of vertices to select and number~$\ell$ of layers to select.
\end{theorem}

\begin{proof}
We give a parameterized reduction from
\textsc{Biclique} which, given an undirected graph~$H$ and a positive integer~$h$,
asks whether $H$~contains a $2h$-vertex \emph{biclique}, a complete bipartite subgraph~$K_{h, h}$ in which both partite sets have size~$h$. 

\emph{Reduction idea.}\quad The main idea is to construct a layer for each vertex $v$ of the input graph using Algorithm $A_\Pi$ from \autoref{def:staggered} such that the optional vertices correspond to the open neighborhood of $v$ in the input graph. Then we can show that selecting layers corresponds to selecting one half of the biclique and the optional vertices that are part of the selected subgraph form the second half of the biclique.

Let~$(H=(U, F),h)$ be an instance of \textsc{Biclique} and assume $h \ge 2$ without loss of generality.
Let $A_{\Pi}, f, g$ be the algorithm and functions promised by the definition of being $(f,g)$-staggered. We construct an instance of \PML in the following way.

For all $v \in U$, let $N_H(v)$ be the neighborhood of $v$ with respect to $H$.
Run Algorithm~$A_{\Pi}$ (see \autoref{def:staggered}) on input $(U, N_H(v), h)$ to create graphs $G_v$ for each $v
\in U$. A graph~$G_v$ created that way has a vertex set $V_v = \bigcup_{u\in U} X_u \cup Y$. For all created graphs, we identify the vertices of the sets $X_u$ in an arbitrary but fixed fashion, the same for vertices in $Y$. This allows us to say that all graphs $G_v$ are defined over the same vertex set. Set $k := h\cdot f(h)+g(h)$ and~$\ell := h$. Note that the parameters~$k$ and~$\ell$ only depend on the solution size~$h$ of \textsc{Biclique}. By~\autoref{def:staggered}, this procedure runs in polynomial time. Now we show that
$(\{G_v \mid v \in U\}, k, \ell)$ is a yes-instance of \PML if and only if $(H, h)$ is
a yes-instance of \textsc{Biclique}.

$(\Rightarrow)$: Assume that $(H, h)$ is a yes-instance of \textsc{Biclique} and
let $(C, D)$ with $C, D \subseteq U$ and~$|C|=|D|=h$ represent a biclique. Then we set $X :=
\bigcup_{v\in C} X_v \cup Y$, where $Y$ and the~$X_v$ are the vertex sets promised by Conditions~1 and~2 of \autoref{def:staggered}. Note that $|X| = h\cdot f(h)+g(h) = k$. Furthermore, for all $v' \in D$ and all $v\in C$ such that
$X_v \subset X$, it holds that $v\in N_H(v')$. Hence, by Condition~3 of \autoref{def:staggered}, $G_{v'}[X]
\in \Pi$. It follows that the number of layers~$G_v$ with~$G_{v}[X] \in \Pi$ is at least~$h = \ell$. Consequently,~$(X, D)$ is a solution of \PML.

 $(\Leftarrow)$: Assume that $(\{G_v \mid v \in U\}, k, \ell)$ is a yes-instance of \PML.
Then we know that there are graphs $G_i$, with $i \in L$, $L \subseteq U$, $|L| \geq h$,
and a vertex set $X \subseteq V$ with $|X| \geq k$, such that $G_i[X] \in \Pi$
for all $i \in L$. By the construction of $G_i$ (Conditions~1,~2, and~3 of \autoref{def:staggered}), we know that $X =
\bigcup_{v\in W'} X_v \cup Y$ for some $W' \subseteq U$ with $|W'|\geq h$.
Furthermore, we know that if $i \in L$ then for all $j\in W'$ (that is $X_j
\subset X$) we have that $i$ is a neighbor of $j$. Lastly, $i\in L$ implies that $X_i
\not\subset X$ and hence $i\notin W'$. Hence, we have that $(L, W')$ is a
biclique in $H$ with $|L| \geq h$ and $|W'|\geq h$.
\end{proof}

In the following corollary, we present several graph properties~$\Pi$ for which the single-layer case $\Pi$-\textsc{Subgraph} is
polynomial-time solvable but by application of \autoref{thm:meta} \PML{} is \NP-hard and \W{1}-hard
when parameterized by the combined parameter number~$k$ of vertices to select and number~$\ell$ of subgraph layers.  The proof of \autoref{cor:hardness} simply consists of the description of Algorithm~$A_{\Pi}$ of \autoref{thm:meta} for those graph properties. For the definitions of the graph properties appearing in \autoref{cor:hardness} see Appendix~\ref{appendix:definitions}.
\begin{cor}
\label{cor:hardness}
\PML is \NP-hard and \W{1}-hard when parameterized by the 
combined parameter number~$k$ of vertices to select and number~$\ell$ of layers to select for $\Pi \in
\{$\gprop{$c$-Con\-nectivity}, \gprop{$c$-Core}, \gprop{$c$-Factor}, \gprop{Con\-nectivity},  \gprop{$c$-Truss}, \gprop{Hamiltonian}, \gprop{Matching}, \gprop{Star}, \gprop{Tree}$\}$.
\end{cor}

\begin{proof}
For each of the listed properties~$\Pi$, we show that $\Pi$ is~$(f,g)$-staggered for some functions~$f$ and~$g$. To this end, we describe the polynomial-time algorithm~$A_\Pi$ with inputs~$W$, $W'\subseteq W$, and~$\alpha$. \autoref{thm:meta} then yields \NP-hardness and~\W{1}-hardness of \PML{} when parameterized by the combined parameter number~$k$ of vertices to select and number~$\ell$ of subgraph layers. \autoref{fig:ccore} shows the graph constructed by Algorithm~$A_{\Pi}$ for $\Pi \in\{$\gprop{$c$-Core}, \gprop{$c$-Connectivity}$\}$ with $c=3$.
\begin{figure}[t]
\begin{center}
    \begin{tikzpicture}[ scale=1.5]

\node at (0,1.5) {obligatory vertices $Y$};
\node at (0,-.5) {optional vertices $\bigcup_{v\in W'} X_v$};
\node at (0,-2.2) {forbidden vertices $\bigcup_{v\in W\setminus W'} X_v$};

      \draw[rounded corners,dashed,color=gray, fill=white!50!green] (-2.5,0.7) rectangle (2.5,1.3);
      \draw[rounded corners,dashed,color=gray, fill=white!50!yellow] (-2.5,-.3) rectangle (2.5,0.3);
      \draw[rounded corners,dashed,color=gray, fill=white!60!red] (-2.5,-2) rectangle (2.5,-.9);

    \node[vert] (o1) at (-1,1) {};   
    \node[vert] (o2) at (0,1) {};   
    \node[vert] (o3) at (1,1) {};
    \node[vert] (n1) at (-2,0) {}; 
    \node[vert] (n2) at (-1,0) {}; 
    \node[vert] (n3) at (0,0) {}; 
    \node[vert] (n4) at (1,0) {};
    \node[vert] (n5) at (2,0) {}; 
    \node[vert] (v1) at (-2,-1.2) {}; 
    \node[vert] (v2) at (-1,-1.2) {};
    \node[vert] (v3) at (0,-1.2) {};
    \node[vert] (v4) at (1,-1.2) {};
    \node[vert] (v5) at (2,-1.2) {};
    \node[vert] (v6) at (-1.5,-1.7) {}; 
    \node[vert] (v7) at (-.5,-1.7) {};
    \node[vert] (v8) at (0.5,-1.7) {};
    \node[vert] (v9) at (1.5,-1.7) {};
    
    \draw (o1) -- (n1);
    \draw (o1) -- (n2);
    \draw (o1) -- (n3);
    \draw (o1) -- (n4);
    \draw (o1) -- (n5);
    \draw (o2) -- (n1);
    \draw (o2) -- (n2);
    \draw (o2) -- (n3);
    \draw (o2) -- (n4);
    \draw (o2) -- (n5);
    \draw (o3) -- (n1);
    \draw (o3) -- (n2);
    \draw (o3) -- (n3);
    \draw (o3) -- (n4);
    \draw (o3) -- (n5);


    \end{tikzpicture}
    \end{center}
    \caption{Visualization of graph constructed by Algorithm~$A_{\Pi}$ for $\Pi \in\{$\gprop{$c$-Core}, \gprop{$c$-Connectivity}$\}$ with $c=3$.}
    \label{fig:ccore}
\end{figure}

\noindent \emph{Algorithm~$A_{\Pi}$ for $\Pi \in\{$\gprop{Connectivity}, \gprop{Tree}, \gprop{Star}, \gprop{$1$-Core}$\}$:}
  
  We construct the graph $G=(V,E)$ as follows. Let $X_w:=\{w\}$ for all $w\in W$ and $Y:=\{u\}$ and hence $V=W\cup\{u\}$. Add an edge $\{u, w\}$ to the edge set~$E$ for each vertex $w
  \in W'$. This clearly fulfills Conditions~1 and~2 from \autoref{def:staggered} for $f(\alpha)=g(\alpha)=1$. Note that the graph $G$ contains a star with~$u$ as center and all vertices $w\in W'$ as leaves and all other vertices are isolated. Furthermore, it is easy to see that no graph that satisfies one of the properties \gprop{Connectivity}, \gprop{Tree}, \gprop{Star}, and \gprop{$1$-Core} can contain an isolated vertex. Hence, each subgraph of~$G$ that satisfies one of those properties is a subtree of the star in~$G$ and therefore has to contain the center~$u$. It follows that Condition~3 from \autoref{def:staggered} is fulfilled. 
  
\noindent \emph{Algorithm~$A_{\Pi}$ for $\Pi \in\{$\gprop{$c$-Core}, \gprop{$c$-Connectivity}$\}$ with $c>1$:}
  
  We construct the graph $G=(V,E)$ as follows. Let $X_w:=\{w\}$ for all $w\in W$ and $Y:=\{u_1, \ldots, u_c\}$ and hence $V=W\cup\{u_1, \ldots, u_c\}$. Add edges $\{u, w\}$ to~$E$ for all vertices $u \in Y$ and $w \in W'$. This clearly fulfills Conditions~1 and~2 from \autoref{def:staggered} for $f(\alpha)=1$ and $g(\alpha)=c$. Note that the graph $G$ contains a complete bipartite subgraph where one part is~$Y$ and the other is~$W'$ and all other vertices are isolated. Furthermore, it is easy to see that no graph that satisfies one of the properties \gprop{$c$-Core} and \gprop{$c$-Connectivity} can contain an isolated vertex. It follows that every subgraph of $G$ that satisfies one of those properties has to be a subgraph of the complete bipartite subgraph in~$G$. In order to be a $c$-core or $c$-connected, a complete bipartite graph has to have at least~$c$~vertices in each part. Since we have that $|Y|=c$, all vertices of $Y$ have to be contained in each subgraph of $G$ satisfying one of the properties \gprop{$c$-Core} and \gprop{$c$-Connectivity} and we get that Condition~3 from \autoref{def:staggered} is fulfilled. The resulting graph is visualized in \autoref{fig:ccore}.
  
\noindent \emph{Algorithm~$A_{\Pi}$ for $\Pi = c\gprop{-Truss}$:}
  
  We construct the graph $G=(V,E)$ as follows. Let $X_w:=\{w\}$ for all $w\in W$ and $Y:=\{u_1, \ldots, u_{c-1}\}$ and hence $V=W\cup\{u_1, \ldots, u_{c-1}\}$. Add edges $\{u, w\}$ to~$E$ for all vertices $u \in Y$ and $w \in W' \cup Y$ with $u \neq w$. This clearly fulfills Conditions~1 and~2 from \autoref{def:staggered} for $f(\alpha)=1$ and~$g(\alpha)=c-1$. Note that the vertices in~$Y$ form a clique of size~$c-1$. Furthermore, every vertex from $W'$ is connected exactly to all vertices in~$Y$ and hence every triangle in~$G$ contains at least two vertices from~$Y$ and for every $w\in W'$ we have that $Y\cup\{w\}$ forms a clique of size~$c$ and hence a $c$-truss. It follows that for each $W''\subseteq W'$ with $W''\neq \emptyset$ we have that $G[W''\cup Y]$ is a $c$-truss. Since vertices in $W'$ are not connected to each other, all vertices in $Y$ are necessary to produce enough triangles for the edges going from $W''$ to $Y$. Also, we cannot add any isolated vertices to the subgraph since we require $c$-trusses to be connected. It follows that Condition~3 from \autoref{def:staggered} is fulfilled.
  
\noindent \emph{Algorithm~$A_{\Pi}$ for $\Pi \in\{$\gprop{Matching}, \gprop{$c$-Factor}$\}$:}

    Recall that $\gprop{Matching}=\gprop{$1$-Factor}$. Hence, we describe the algorithm only for the more general \gprop{$c$-Factor} property.
  We construct the graph $G=(V,E)$ as follows. (For \gprop{Matching} set $c=1$.) Let $X_w:=\{w_1, \ldots, w_{c+1}\}$ for all $w\in W$ and~$Y:=\emptyset$.
  Add all edges $\{w_i, w_j\}$ with $1\le i < j \le c+1$ to $E$ for all $w\in W'$.
  This clearly fulfills Conditions~1 and~2 from \autoref{def:staggered} for $f(\alpha)=c+1$ and $g(\alpha)=0$. Note that for each $w\in W'$, $G[X_v]$ is a complete graph of size $c+1$ and hence a connected $c$-regular graph. Furthermore, for any set $X'\subset X_w$ and any set $X'' \subseteq V \setminus X_w$, we have that $G[X']$ is a connected component of $G[X'\cup X'']$ with a minimum degree strictly smaller than $c$. Hence, $G[X'\cup X'']$ does not have a $c$-factor. Together with the fact that all vertices in the sets~$X_w$ with $w\in W\setminus W'$ are isolated, it follows that if there is a vertex set $X$ such that $G[X]$ has a $c$-factor, then $X = \bigcup_{w\in W''}X_v$ for some~$W''\subseteq W'$. Hence, we get that Condition~3 from \autoref{def:staggered} is fulfilled. 
  
\noindent \emph{Algorithm~$A_{\Pi}$ for $\Pi = \gprop{Hamiltonian}$:}
  
  We construct the graph $G=(V,E)$ as follows. Let $X_w:=\{w\}$ for all $w\in W$ and $Y:=\emptyset$ and hence $V=W$. Add edges $\{u, v\}$ to~$E$ for all vertices $u, v \in W'$. This clearly fulfills Conditions~1 and~2 from \autoref{def:staggered} for $f(\alpha)=1$ and $g(\alpha)=0$. Note that the vertices $W'$ form a clique in~$G$ and all other vertices are isolated. Furthermore, we have that any clique is also a Hamiltonian subgraph and any Hamiltonian subgraph cannot contain any isolated vertices. Hence, we get that Condition~3 from \autoref{def:staggered} is fulfilled.

The results follow from the existence of the described algorithms.
\end{proof}
A particular consequence of \autoref{cor:hardness} is that, while the polynomial-time solvability of the
connected component detection algorithm for two layers by \citet{GHPR03} generalizes to any constant number of layers (\autoref{thm:partitioning-properties}), it does not generalize to an arbitrary, given number of layers.

\section{Matching and $c$-Factors}\label{sec:mlmatching}

In this section we consider the graph properties~\gprop{Matching} and
its generalization~\gprop{$c$-Factor}. To recall, a
graph~$G$ has property \gprop{Matching} if it contains a perfect
matching and $G$ has property \gprop{$c$-Factor} if it has a
$c$-regular subgraph containing all vertices of~$G$. Finding a
maximum $c$\nobreakdash-factor for a given graph is polynomial-time solvable for all~$c$ (see
\citet{Plu07} or \citet[Chapter~3]{Nic14}, for an overview on graph factors).
For these properties, \autoref{thm:meta}
shows that \PML\ is \W{1}-hard with respect to the number~$k$ of vertices and the
number~$\ell$ of layers to select. \autoref{thm:meta} does
not rule out fixed-parameter tractability with respect to the total
number~$t$ of layers, or with respect to the number~$k$ of vertices to
select if the number~$\ell$ of layers to select is constant. In this section, however, we show
that both \MatchML and \gprop{$c$-Factor}-\ML{} are hard even for a
constant number of layers, thus strengthening the statement of
\autoref{thm:meta} for these properties.

\paragraph{Matching.} As mentioned, through \autoref{thm:meta} we get in \autoref{cor:hardness} in \autoref{sec:meta} that \MatchML is
\W{1}-hard when parameterized by the combined parameter number~$k$ of vertices to select
and number~$\ell$ of layers to select.
Through closer inspection we can get a stronger result. We show that \MatchML is polynomial-time solvable
for~$\ell\le2$ and becomes \W{1}-hard when parameterized by the number~$k$ of vertices to select already for~$\ell\ge3$. Intuitively, the reason for the computational
complexity transition from two layers to three layers is as follows. 
By overlaying two matchings one may create cycles and paths
but without connections between them. We can cope with this by
finding a maximum weighted matching in an auxiliary graph. Adding a
third layer, however, allows arbitrary connections between cycles and
paths, which allows the construction of gadgets to show hardness.

We now reduce \gprop{Matching}-\ML\ with $\ell=2$ layers to
\textsc{Maximum Weight Matching}, that is, to the problem where we are
given a graph with edge weights and ask to find a matching with maximum edge
weights. To this end, let $G_1=(V,E_1)$ and $G_2=(V,E_2)$ be the two
layers of the input graph for which we would like to know whether
there is an $X\subseteq V$ of size at least $k$ such that
both~$G_1[X]$ and~$G_2[X]$ have a perfect matching. The reduction is
given in the following lemma and visualized in \autoref{fig:matching}.

\begin{figure}[t]
\begin{center}
    \begin{tikzpicture}[line width=1pt, scale=1.6]   
    \node[vert] (A2) at (2,1) {}; 
    \node[vert] (B2) at (3,1) {};
    \node[vert] (C2) at (1,0) {};
    \node[vert] (D2) at (2,0) {};
    \node[vert] (E2) at (3,0) {};
    \node[vert] (F2) at (1,-1) {};
    \node[vert] (G2) at (2,-1) {};
    \node (graph1) at (2, -1.5) {$G_2$};
    
    \node[vert] (A1) at (-2,1) {}; 
    \node[vert] (B1) at (-1,1) {};
    \node[vert] (C1) at (-3,0) {};
    \node[vert] (D1) at (-2,0) {};
    \node[vert] (E1) at (-1,0) {};
    \node[vert] (F1) at (-3,-1) {};
    \node[vert] (G1) at (-2,-1) {};
    \node (graph1) at (-2, -1.5) {$G_1$};
    
    \path (A1) edge[dashed,color=gray, bend left] (A2);
    \path (B1) edge[dashed,color=gray, bend left] (B2);
    \path (C1) edge[dashed,color=gray, bend left] (C2);
    \path (D1) edge[dashed,color=gray, bend left] (D2);
    \path (E1) edge[dashed,color=gray, bend left] (E2);
    \path (F1) edge[dashed,color=gray, bend left] (F2);
    
	\path (G1) edge[dashed,color=gray, bend left, line width=3pt] (G2);
    
    \draw (A1) -- (B1);
    \draw (B1) -- (D1);
    \draw (C1) -- (G1);
    \draw (D1) -- (E1);
    \draw (E1) -- (G1);
    
    \path (A2) edge (D2);
    \path (C2) edge (D2);
    \path (D2) edge (G2);
    \path (F2) edge (G2);

    \path (A1) edge[line width=3pt] (C1);
    \path (B1) edge[line width=3pt] (E1);
    \path (D1) edge[line width=3pt] (F1);
    
    \path (A2) edge[line width=3pt] (B2);
    \path (C2) edge[line width=3pt] (F2);
    \path (D2) edge[line width=3pt] (E2);

    \end{tikzpicture}
\end{center}
\caption{Construction of the graph $G'=(V',E')$ from graphs $G_1=(V, E_1)$
and $G_2=(V, E_1)$.
Black edges have weight $|V|+1$ and gray dashed edges have weight $|V|$. The
thick edges are a maximum-weight matching for $G'$.}
\label{fig:matching}
\end{figure}

\begin{lem}\label{lemma:reductiontomaxweightmatching}
Given two graphs $G_1=(V,E_1)$ and $G_2=(V,E_2)$, define a graph $G'=(V',E')$ as
follows:
\begin{compactitem}
 \item $V'=\{v_1,v_2\mid v\in V\}$ and
 \item $E'=\{\{v_1,v_2\}\mid v\in V\}\cup \{\{u_i,v_i\}\mid \{u,v\}\in E_i\}$.
\end{compactitem}
\smallskip
Define a weight function $w\colon E'\to\NN$ as follows; let $n:=|V|$:
\begin{align*}
w(\{u_i,v_j\})=\begin{cases}
               n & \mbox{if $i\neq j$ and $u=v$,}\\
               n+1 & \mbox{if $i=j$ (and $u\neq v$).}
               \end{cases}
\end{align*}
Let $k\in\NN$. Then, there is a set $X\subseteq V$ of size at least $k$ such that
both $G_1[X]$ and $G_2[X]$ have a perfect matching if and only if the graph $G'$ has
a matching of $w$-weight at least $n^2+k$.
\end{lem}

\begin{proof}
Assume that $G'$ has a matching $M'\subseteq E'$ with $w(M')\geq n^2+k$. Let us
first check that $M'$ is in fact a perfect matching: If $|M'|\leq n-1$, then
$w(M')\leq (n+1)(n-1)=n^2-1<n^2+k$, which would contradict the choice of $M'$. Thus $|M'|\geq n$ but then $M'$ must have exactly $n$ edges since $G'$ has $2n$ vertices. That is,~$M'$ is a perfect matching. Now we show how to get a vertex set $X\subseteq V$ such that both $G_1[X]$ and $G_2[X]$ have perfect matchings. Let
\(
Y:=\{v\mid v\in V \wedge \{v_1,v_2\}\in M'\},
\)
that is, $Y$ is the set of vertices of $V$ whose copies in $G'$ are matched
to each other under $M'$. Let $X:= V\setminus Y$. It can be easily checked that both $G_1[X]$ and $G_2[X]$ have perfect matchings; we show this for $G_1[X]$: For any $v\in X$ we know that $\{v_1,v_2\}\notin M'$, or else we would have $v\in Y$ and $v\notin X$. Thus, using that $M'$ is a perfect matching, $v_1$ must be matched to another vertex~$u_1$, which is then also in $X$, by definition. It follows that $M'$ induces a perfect matching on $G'[X_1]$ where $X_1:=\{v_1\mid v\in X\}$. Since $G_1[X]$ is an isomorphic copy of $G'[X_1]$ under the canonical isomorphism $\phi\colon v\mapsto v_1$, we get that $G_1[X]$ also has a perfect matching.

It remains to check that $X$ has size at least $k$: Observe that $|X|+|Y|=n$ since every vertex of $V$ is either in $X$ or in $Y$. Each $v\in Y$ corresponds to a matching edge $\{v_1,v_2\}\in M'$ which has weight $n$ under $w$. Thus, if $|X|<k$, then $|Y|>n-k$, which implies that $w(M')\leq |X|\cdot(n+1)+|Y|n< kn + k + (n-k)n=n^2+k$, contradicting the choice of $M'$.

Assume now that both $G_1[X]$ and $G_2[X]$ have perfect matchings $M_1$ and $M_2$ for some~$X$ of size at least $k$. Clearly, the size of $X$ must be even. Define a matching $M'$ of $G'$ by
\[
M':=\{\{v_1,v_2\}\mid v\in V\setminus X\}\cup\{\{u_i,v_i\} \mid \{u,v\}\in M_i\}.
\]
In other words, $M'$ is obtained by copying $M_1$ and $M_2$ to $G'$ in the obvious way and matching all leftover vertices by the edges between the copies of the same vertex.

Clearly, for each vertex $v\in V\setminus X$ this adds an edge $\{v_1,v_2\}$ of weight $n$ to $M'$. From~$M_1$ and~$M_2$ we copied $\frac{|X|}{2}$ edges each, which results in exactly $|X|\geq k$ edges of weight~$n+1$ in~$M'$. Thus, the total weight of~$M'$ is $n^2+k$, as claimed.
\end{proof}

To show that \MatchML remains \W{1}-hard for any $\ell\ge3$ when parameterized by $k$, we reduce from \MCC which is known to be \W{1}-hard when
parameterized by the solution size~\cite{FHRV09}. 

\begin{theorem}
\label{thm:matching}
\MatchML can be solved in polynomial time if~$\ell\le2$.
It is \NP-hard and \W{1}-hard when parameterized by the number of vertices to select~$k$ for all~$\ell\ge3$ and total numbers of layers~$t\ge\ell$.
\end{theorem}

\begin{proof}

We get the polynomial-time solvability of \MatchML for the case~$\ell\leq 2$ from \autoref{lemma:reductiontomaxweightmatching}. For the case of
$\ell\geq 3$ we give a parameterized reduction from \MCC. 
 In \MCC, we are given an $h$-partite graph $H=(U_1\uplus \ldots\uplus U_h,F)$ and
need to determine whether it contains a clique of size $h$.
Note that such a clique necessarily contains exactly one vertex from each set
$U_i$ and cliques of more than $h$ vertices are impossible. We say that
vertices from $U_i$ have \emph{color}~$i$. 
 
\emph{Reduction idea.}\quad The main idea is to create the first two layers in such a way that selecting a subgraph that admits a perfect matching in both layers corresponds to selecting exactly one vertex of each color. The third layer is constructed in a way that each subgraph that admits perfect matchings in the first two layer admits also a perfect matching in this layer if all selected vertices are pairwise connected and hence form a multicolored clique.

Without loss of
generality, we assume that the number~$h$ of colors is even.
 We construct an instance of \MatchML for $t=\ell=3$ as follows
 and then argue that the construction is easily generalizable to larger numbers of layers.
 
 \emph{Vertices.}\quad
 First, create $h-1$ vertices for each vertex in graph~$H$ (one vertex for each color other than its own color).
 Formally, for each color $1 \le j \le h$ and each~$u_i \in U_j$,
 create the vertex set $V_i$ consisting of the vertices~$v_{(i,j')}$, $j' \in
 (\{1,\dots,h\} \setminus \{j\})$.
 Second, create one \emph{color vertex} $w_j$ for each color~$j \in
 \{1,\dots,h\}$. We denote the set of color vertices as~$W:=\bigcup_{1 \le j \le
 h}\{w_j\}$.

\emph{Vertex selection gadget by graphs~$G_1$
and~$G_2$.}\quad The vertex selection gadget is intended to make sure the vertices selected in any valid solution of the constructed \MatchML instance correspond to pairwise differently colored vertices of~$H$.
The vertex selection
 gadget is visualized in \autoref{fig:matching:selection}.
We construct it as follows. For each color~$1 \le j \le h$, create for each~$u_i \in
U_j$ one cycle on $\{w_j\} \cup V_i$ in the graph~$G_1 \cup G_2$
 such that the edges are alternatingly from $G_1$ and from $G_2$.
 These $|U_j|$~cycles are all of length~$h$ and share only the color vertex~$w_j$.
 To realize this, create the following edges.
 For each $1 \le z \le h-2$, create an edge in graph~$G_{(z\text{ mod }2)+1}$
 between $v_{(i,z)}$ and $v_{(i,z+1)}$ if $z<j-1$,
 between $v_{(i,z+1)}$ and $v_{(i,z+2)}$ if~$z\ge j$, and
 between $v_{(i,z)}$ and $v_{(i,z+2)}$ if $z=j-1$.
 Create an edge between $w_j$ and $v_{(i,1)}$ in graph~$G_2$,
 between $v_{(i,h)}$ and $w_j$ in graph~$G_1$ if $j\neq h$, and
 between $v_{(i,h-1)}$ and $w_j$ in graph~$G_1$ if $j= h$.
 
\begin{figure}[t]
\begin{center}
    \begin{tikzpicture}[line width=1pt, scale=2]

    \node[vert] (v1) at (-2,.75) [label=90:{$v_{(1,1)}$}]{}; 
    \node[vert] (v2) at (-1,.75) [label=90:{$v_{(1,2)}$}]{};
    \node[vert] (v3) at (0,.75) [label=90:{$v_{(1,3)}$}]{};
    \node[vert] (v4) at (1,.75) [label=90:{$v_{(1,5)}$}]{};
    \node[vert] (v5) at (2,.75) [label=90:{$v_{(1,6)}$}]{};
    \node[vert] (c) at (0,0) [label=-90:{$w_4$}]{};
    \node[vert] (w1) at (-2,-.75) [label=-90:{$v_{(2,1)}$}]{}; 
    \node[vert] (w2) at (-1,-.75) [label=-90:{$v_{(2,2)}$}]{};
    \node[vert] (w3) at (0,-.75) [label=-90:{$v_{(2,3)}$}]{};
    \node[vert] (w4) at (1,-.75) [label=-90:{$v_{(2,5)}$}]{};
    \node[vert] (w5) at (2,-.75) [label=-90:{$v_{(2,6)}$}]{};

    
    \path (c) edge (w1);
    \path (w1) edge[dashed,color=red] (w2);
    \path (w2) edge (w3);
    \path (w3) edge[dashed,color=red] (w4);
    \path (w4) edge (w5);
    \path (w5) edge[dashed,color=red] (c);

        \path (c) edge[line width=3pt] (v1);
\path (v1) edge[dashed,color=red,line width=3pt] (v2);
\path (v2) edge[line width=3pt] (v3);
\path (v3) edge[dashed,color=red,line width=3pt] (v4);
\path (v4) edge[line width=3pt] (v5);
\path (v5) edge[dashed,color=red,line width=3pt] (c);
    \end{tikzpicture}
    \end{center}
    \caption{Parts of the vertex selection gadget $G_1$ and $G_2$ for two
    vertices~$u_1, u_2 \in U_4$, where the number of colors is~$h=6$. Black edges belong to~$E_1$ and red
    dashed edges belong to~$E_2$. The thick and the thin edges both create
    cycles that have edges alternating between Layers~1 and~2. Since the
    color vertex~$w_4$ is contained in both cycles, only one of these cycles
    can be contained in a matching subgraph.}
    \label{fig:matching:selection}
\end{figure}

\emph{Validation gadget by graph $G_3$.}\quad The validation gadget is intended to make sure the vertices selected in any valid solution of the constructed \MatchML instance correspond to a clique in~$H$.
The validation gadget is visualized in \autoref{fig:matching:validation} and \autoref{fig:matching:validation2}. It is constructed as follows.
For each vertex pair $u_i$, $u_{i'}$ with $u_i \in U_j$ and $u_{i'}
\in U_{j'}$ such that $u_i$ and $u_{i'}$ are adjacent in~$H$, we create an edge between $v_{(i, j')}$ and $v_{(i', j)}$ in $G_3$.
Furthermore, create the edge~$\{w_j,w_{j+h/2}\}$ for each $1 \le j \le h/2$. 

\begin{figure}[t]
\begin{center}
    \begin{tikzpicture}[line width=1pt, scale=2]

    \node[vert] (c1) at (-1,0.75) [label=0:{$w_1$}]{};
    \node[vert] (c4) at (-1,0.25) [label=0:{$w_4$}]{};    
    \node[vert] (c2) at (0,0.75) [label=0:{$w_2$}]{};
    \node[vert] (c5) at (0,0.25) [label=0:{$w_5$}]{};    
    \node[vert] (c3) at (1,0.75) [label=0:{$w_3$}]{};
    \node[vert] (c6) at (1,0.25) [label=0:{$w_6$}]{};
    \node[vert] (v1) at (-2,-.5) [label=90:{$v_{(1,1)}$}]{}; 
    \node[vert] (v2) at (-1,-.5) [label=90:{$v_{(1,2)}$}]{};
    \node[vert] (v3) at (0,-.5) [label=90:{$v_{(1,3)}$}]{};
    \node[vert] (v4) at (1,-.5) [label=90:{$v_{(1,5)}$}]{};
    \node[vert] (v5) at (2,-.5) [label=90:{$v_{(1,6)}$}]{};

    \node[vert] (w1) at (-1.75,-1.25) {};
    \node[vert] (w2) at (-1.25,-1.25) {};
    \node[vert] (w3) at (-.75,-1.25) {};
    \node[vert] (w4) at (-.25,-1.25) {};
    
    \draw (v2) -- (w1);
    \draw (v2) -- (w2);
    \draw (v2) -- (w3);
    \draw (v2) -- (w4);
    \draw (c1) -- (c4);
    \draw (c2) -- (c5);
    \draw (c3) -- (c6);

                        \draw [
    thick,
    decoration={
        brace,
        raise=.4cm
    },
    decorate
] (w4) -- (w1)
node [pos=0.5,anchor=north,yshift=-.5cm] {\small $\{v_{(i, 4)}\mid u_i\in N_H(u_1)\cap U_2$\}};

    \end{tikzpicture}
    \end{center}
    \caption{Parts of the validation gadget $G_3$ for $h=6$ and a vertex $u_1\in U_4$.
     The color vertices $w_1$ to $w_6$ form a matching and each vertex $v_{(1, j)}$ is
     connected to all vertices~$v_{(i, 4)}$ with $u_i\in N_H(u_1)\cap U_j$, as exemplarily
     visualized for $v_{(1, 2)}$.}
    \label{fig:matching:validation}
\end{figure}

\begin{figure}[t]
\begin{center}
    \begin{tikzpicture}[line width=1pt, scale=1.7]

    \node[vert] (u1) at (-.5,1.35) [label=90:{$u_1\in U_1$}]{};
    \node[vert] (u2) at (-1.5,-.15) [label=-90:{$u_2\in U_2$}]{};    
    \node[vert] (u3) at (.5,-.15) [label=1-90:{$u_3\in U_3$}]{};
    
    \draw (u1) -- (u2);
    \draw (u1) -- (u3);
    \draw (u2) -- (u3);
    
    \node[vert] (v12) at (3,1.5) [label=90:{$v_{(1,2)}$}]{};
    \node[vert] (v13) at (3.5,1.5) [label=90:{$v_{(1,3)}$}]{};
    \node[vert] (v14) at (4,1.5) [label=90:{$v_{(1,4)}$}]{};
    \node[vert] (v21) at (2.2,.3) [label=-100:{$v_{(2,1)}$}]{};
    \node[vert] (v23) at (2.5,0) [label=-100:{$v_{(2,3)}$}]{}; 
    \node[vert] (v24) at (2.8,-.3) [label=-100:{$v_{(2,4)}$}]{};     
    \node[vert] (v31) at (4.2,-.3) [label=1-80:{$v_{(3,1)}$}]{};    
    \node[vert] (v32) at (4.5,0) [label=1-80:{$v_{(3,2)}$}]{};    
    \node[vert] (v34) at (4.8,.3) [label=1-80:{$v_{(3,4)}$}]{};
    
    \draw (v12) -- (v21);
    \draw (v13) -- (v31);
    \draw (v23) -- (v32);
    
    \draw[dashed] (1.25,2) -- (1.25,-.8);
    
    \end{tikzpicture}
    \end{center}
    \caption{Parts of the validation gadget $G_3$ for $h=4$ and a triangle consisting of~$u_1\in U_1$, $u_2\in U_2$, and $u_3\in U_3$. The triangle (to be interpreted as a part of the original graph) is depicted on the left, the corresponding part of the gadget is depicted on the right.
     The color vertices $w_1$ to $w_4$ are omitted.}
    \label{fig:matching:validation2}
\end{figure}

 Finally, by setting~$k=h^2$ and~$t=\ell=3$ we complete the construction,
 which can clearly be performed in polynomial time and the new parameter~$k$
 solely depends on~$h$.

\emph{Correctness.}\quad  It remains to show that graph~$H$ has a clique that contains each color exactly once if and
 only if there is a vertex set~$X\subseteq V$ with $|X|\geq k$ such that
 graph~$G_z[X]$ contains a perfect matching for each~$1 \le z \le 3$.
 
 $(\Rightarrow)$: 
 Assume that graph~$H$ has a clique
 $K:=\{u_1,u_2,\dots,u_h\}$ and, without loss of generality, $u_i \in U_i$ for
 all $1\le i\le h$.
 We show that $X:=W \cup V_1 \cup V_2 \cup \dots \cup V_h$ is
 a solution for our \MatchML instance.
 By construction, $X$~is of size $h+ h \cdot (h-1)=h^2 = k$.
 It remains to show that graph~$G_z[X]$ has a perfect matching for each $1 \le z \le 3$.
 Recall that we created for each color~$1 \le j \le h$ and for each vertex~$u_i
 \in U_j$ one cycle on $\{w_j\} \cup V_i$ in the graph~$G_1 \cup G_2$ such
 that the edges alternatingly are from $G_1$ and from $G_2$.
 Since $X$~only contains one of these cycles for each color,
 a perfect matching is easy to find for graph~$G_1[X]$ and for graph~$G_2[X]$.
 For graph~$G_3[X]$, we can find the matching~$\{ \{v_{(i,j)},v_{(j,i)}\} \mid (1 \le i,j \le h) \wedge (i \neq j)\} \cup \{
 \{w_j,w_{j+h/2}\} \mid 1 \le j \le h/2\}$, since, by construction, $v_{(i,j)}$ is
 adjacent to $v_{(j,i)}$ if $u_i$~is adjacent to
 $u_j$, and $u_i \in U_i$ and $u_j \in U_j$ (which is the case since~$K$ is a clique).

$(\Leftarrow)$: 
 Assume that there is a vertex set~$X\subseteq V$ with $|X|\geq k$ such that
 graph~$G_z[X]$ contains a perfect matching for each~$1 \le z \le 3$.
 First, consider the graph $G_1 \cup G_2$ and some pair of vertices~$\{x_1, x_2\} \subseteq X$ that is matched in~$G_1$ or in~$G_2$.
 Then, these two vertices must be from the same cycle~$\{w_j\} \cup V_i$ for
 some~$1 \le j \le h$ and~$u_i \in U_j$, since otherwise there is no edge
 between them in any of the two graphs.
 Furthermore, if two vertices from~$\{w_j\} \cup V_i$ are in~$X$, then
 all vertices from~$\{w_j\} \cup V_i$ must be in~$X$
 because otherwise neither~$G_1$ nor~$G_2$ has a perfect matching:
 Every vertex except~$w_j$ has exactly one neighbor in~$G_1$ and one
 neighbor in~$G_2$ which are both enforced to be also contained in~$X$---this enforces the whole
 cycle~$\{w_j\} \cup V_i$ to be contained in~$X$.
 However, $X$~contains the vertices from $\{w_j\} \cup V_i$ for at most
 one~$i$ for every color~$1 \le j \le h$, because $w_j$ can only be matched to
 one vertex in~$G_1$ and to one vertex in~$G_2$.
 This implies that $X$ contains for each~$1 \le j \le h$ all vertices from $\{w_j\} \cup
 V_i$ for exactly one~$i$, since $|X|\ge k=h^2$ and $|\{w_j\} \cup
 V_i|=h$. Without loss of generality, let $V_i\subseteq X$ for all $1\leq i
 \leq h$ and let $u_i$ be the vertex in graph $H$ corresponding to~$V_i$. We
 show that $K=\{u_1,u_2,\dots,u_h\}$ is a clique in $H$.
 In graph~$G_3[X]$ each color vertex~$w_j$ must be matched to
 its only neighbor:~$w_{j+h/2}$ if~$j\leq h/2$ and~$w_{j-h/2}$ if~$j>
 h/2$ (and cannot be matched to $v_{(i, j)}$-vertices).
 Let~$\{u_i,u_{i'}\}\subseteq K$, with $u_i \in U_j$ and $u_{i'} \in U_{j'}$. 
 Then, note that~$v_{(i,j')}$ and~$v_{(i', j)}$ must be matched since, by
 construction of~$G_3$, vertex~$v_{(i,j')}$ is only adjacent to vertex~$v_{(i',
 j)}$.
 Moreover, if~$v_{(i,j')}$ is adjacent to~$v_{(i', j)}$ in~$G_3$, then $u_i$ is
 adjacent to~$u_{i'}$ in~$H$. 
 Thus overall, $K$ is a $h$-vertex clique in~$H$, as required.
 
 To make
 this reduction work for any $t \geq \ell > 3$, we insert $\ell-3$
 additional layers of complete graphs and $t-\ell$ layers of edgeless graphs.
 \end{proof}
 
 \paragraph{$c$-Factors.} We now show that \gprop{$c$-Factor}-\ML{} with $c\ge 2$ is \W{1}-hard when parameterized by $k$
for $\ell\ge2$. We reduce from \MCC which is known to be \W{1}-hard when
parameterized by the solution size~\cite{FHRV09}. The hardness reduction is similar to the one we use in the proof of \autoref{thm:matching}. Intuitively, since there are connected $c$-regular graphs for $c\ge 2$, we only need one layer to build a vertex selection gadget, whereas in the $1$-factor (matching) case we need two layers.

\needspace{3\baselineskip}
 \begin{theorem}\label{thm:cfactor}
 For $c\geq 2$, \FactorML is \NP-hard and \W{1}-hard
when parameterized by the number~$k$ of layers to select for all $\ell\geq2$ and all $t\ge\ell$.
 \end{theorem}

 \begin{proof}
 In the following, we prove that, for $c\geq 2$, $c$-\gprop{Factor}-\ML is \W{1}-hard when parameterized
by~$k$ for $t=\ell=2$ and then argue that the construction is easily
generalizable.
As in the proof of \autoref{thm:matching}, we give a parameterized reduction from \MCC: we are given an
$h$-partite graph $H=(U_1\uplus\ldots\uplus U_h,F)$ and need to determine whether it contains a clique of size $h$.
 
\emph{Reduction idea.}\quad The idea of the reduction is similar to the one for \gprop{Matching} in the proof of \autoref{thm:matching}
but already works for $\ell=2$ in the \gprop{$c$-factor} case ($c\geq2$). The
reason that we only need two layers is that we can construct connected
$c$-regular graphs for $c\geq2$ of almost arbitrary size and this allows us to construct a vertex selection gadget with only one layer.

We say that
vertices from $U_i$ have \emph{color} $i$. Without loss of generality, assume that
$h\geq c+1$, that $h$ is even, and that $|F|\ge 3$.

\emph{Vertices.}\quad
For each color~$j$, $1\leq j\leq h$, we do
the following: We create a color vertex~$w_j$ and for each vertex $u_i$ in $U_j$, we
create a set of $h-1$ vertices $V_i = \{v_{(i,j')} \ | \ 1\leq j' \leq h \text{
and } j' \neq j\}$. Note that we have one vertex in $V_i$ for each color except
the color of $u_i$. Let~$W = \{w_j \ | \ 1\leq j \leq h\}$.
Furthermore, for each edge $f\in F$ we create a set of vertices $V_f$ with
$|V_f|=c-1$. Let $V_F = \bigcup_{f\in F}V_f$ and $V = \bigcup_i V_i \cup W
\cup V_F$.

\emph{Vertex selection gadget by graph~$G_1$.}\quad The vertex selection gadget is intended to make sure the vertices selected in any valid solution of the constructed \MatchML instance correspond to pairwise differently colored vertices of~$H$. The vertex selection
 gadget is visualized in \autoref{fig:cfactor:selection}. It is constructed as follows. For
every $1\leq j\leq h$ and every $u_i$ in $U_j$ we do the following: We create a
connected $c$-regular graph on the vertex set $V_i\cup\{w_j\}$. Note that this
can be done as follows: We order the vertices in $V_i\cup\{w_j\}$
arbitrarily and connect each vertex to the $\lfloor c/2 \rfloor$ subsequent vertices,
wrapping around at the end. If $c$ is odd, then we additionally connect each vertex $v$
with the vertex at position $(x+h/2) \text{ mod } h$ in the ordering, where $x$
is the position of vertex $v$.
Furthermore, we create a complete graph on the vertices in $V_F$. 

\begin{figure}[t]
\begin{center}
    \begin{tikzpicture}[line width=1pt, scale=2]

    \node[vert] (v1) at (-2,.75) [label=90:{$v_{(1,1)}$}]{}; 
    \node[vert] (v2) at (-1,1.25) [label=90:{$v_{(1,2)}$}]{};
    \node[vert] (v3) at (0,1.25) [label=90:{$v_{(1,3)}$}]{};
    \node[vert] (v4) at (1,1.25) [label=90:{$v_{(1,5)}$}]{};
    \node[vert] (v5) at (2,.75) [label=90:{$v_{(1,6)}$}]{};
    \node[vert] (c) at (0,0) [label=0:{$w_4$}]{};
    \node[vert] (w1) at (-2,-.75) [label=-90:{$v_{(2,1)}$}]{}; 
    \node[vert] (w2) at (-1,-1.25) [label=-90:{$v_{(2,2)}$}]{};
    \node[vert] (w3) at (0,-1.25) [label=-90:{$v_{(2,3)}$}]{};
    \node[vert] (w4) at (1,-1.25) [label=-90:{$v_{(2,5)}$}]{};
    \node[vert] (w5) at (2,-.75) [label=-90:{$v_{(2,6)}$}]{};

    
    \path (c) edge[dashed] (w1);
    \path (w1) edge[dashed] (w2);
    \path (w2) edge[dashed] (w3);
    \path (w3) edge[dashed] (w4);
    \path (w4) edge[dashed] (w5);
    \path (w5) edge[dashed] (c);
    \path (w1) edge[dashed] (w5);
    \path (c) edge[dashed] (w3);
    \path (w2) edge[dashed, bend left] (w4);

        \path (c) edge (v1);
\path (v1) edge (v2);
\path (v2) edge (v3);
\path (v3) edge (v4);
\path (v4) edge (v5);
\path (v5) edge (c);
    \path (v1) edge (v5);
    \path (c) edge (v3);
    \path (v2) edge[bend right] (v4);
    \end{tikzpicture}
    \end{center}
    \caption{Parts of the vertex selection gadget $G_1$ for two
    vertices~$u_1, u_2 \in U_4$, where the number of colors is~$h=6$ and $c=3$. The normal and the dashed edges both create
    $c$-regular subgraphs. Since the
    color vertex~$w_4$ is contained in both subgraphs, only one of these subgraphs
    can be selected. The complete subgraph on $V_F$ is not depicted.}
    \label{fig:cfactor:selection}
\end{figure}

\emph{Validation gadget by graph~$G_2$.}\quad The validation gadget is intended to make sure the vertices selected in any valid solution of the constructed \MatchML instance correspond to vertices of~$H$ that form a clique in~$H$. The validation
 gadget is visualized in \autoref{fig:cfactor:validation}. It is constructed as follows. For each edge
$f\in F$, we do the following: Let $u_i$ and $u_{i'}$ be the endpoints of~$f$ and $u_i
\in U_j$ and $u_{i'} \in U_{j'}$.  We create a complete graph on the vertices in
$V_f \cup \{v_{(i, j')}, v_{(i', j)}\}$. Note that this complete graph has order
$c+1$ and hence is a $c$-regular graph. Furthermore, we create a complete graph
on all vertices in $W$. 

\begin{figure}[t]
\begin{center}
    \begin{tikzpicture}[line width=1pt, scale=2]
    \node[vert] (v1) at (-2,-.5) [label=90:{$v_{(1,1)}$}]{}; 
    \node[vert] (v2) at (-1,-.5) [label=90:{$v_{(1,2)}$}]{};
    \node[vert] (v3) at (0,-.5) [label=90:{$v_{(1,3)}$}]{};
    \node[vert] (v4) at (1,-.5) [label=90:{$v_{(1,5)}$}]{};
    \node[vert] (v5) at (2,-.5) [label=90:{$v_{(1,6)}$}]{};

    \node[draw, circle, inner sep=2pt] (w11) at (-1.6,-.9) {};
    \node[draw, circle, inner sep=2pt] (w21) at (-1.2,-.9) {};
    \node[draw, circle, inner sep=2pt] (w31) at (-.8,-.9) {};
    \node[draw, circle, inner sep=2pt] (w41) at (-.4,-.9) {};
    
    \node[draw, circle, inner sep=2pt] (w12) at (-1.9,-1.1) {};
    \node[draw, circle, inner sep=2pt] (w22) at (-1.3,-1.1) {};
    \node[draw, circle, inner sep=2pt] (w32) at (-.7,-1.1) {};
    \node[draw, circle, inner sep=2pt] (w42) at (-.1,-1.1) {};

    \node[vert] (w1) at (-2.5,-1.5) {};
    \node[vert] (w2) at (-1.5,-1.5) {};
    \node[vert] (w3) at (-.5,-1.5) {};
    \node[vert] (w4) at (.5,-1.5) {};
    
    \draw (w11) -- (w12);
    \draw (w21) -- (w22);
    \draw (w31) -- (w32);
    \draw (w41) -- (w42);
    
    \draw (w11) -- (v2);
    \draw (w21) -- (v2);
    \draw (w31) -- (v2);
    \draw (w41) -- (v2);
    
    \draw (w1) -- (w12);
    \draw (w2) -- (w22);
    \draw (w3) -- (w32);
    \draw (w4) -- (w42);

    \path (v2) edge[bend angle= 20, bend left] (w12);
    \path (v2) edge[bend angle= 20, bend left] (w22);
    \path (v2) edge[bend angle= 20, bend right] (w32);
    \path (v2) edge[bend angle= 20, bend right] (w42);
    
    \path (w11) edge[bend angle= 20, bend left] (w1);
    \path (w21) edge[bend angle= 20, bend right] (w2);
    \path (w31) edge[bend angle= 20, bend left] (w3);
    \path (w41) edge[bend angle= 20, bend right] (w4);
    
    \path (v2) edge[bend angle= 15, bend right] (w1);
    \path (v2) edge[bend angle= 20, bend left] (w2);
    \path (v2) edge[bend angle= 20, bend right] (w3);
    \path (v2) edge[bend angle= 15, bend left] (w4);

                        \draw [
    thick,
    decoration={
        brace,
        raise=.4cm
    },
    decorate
] (w4) -- (w1)
node [pos=0.5,anchor=north,yshift=-.5cm] {\small $\{v_{(i, 4)}\mid u_i\in N_H(u_1)\cap U_2$\}};

    \end{tikzpicture}
    \end{center}
    \caption{Parts of the validation gadget $G_2$ for $c=3$, $h=6$, and a vertex $u_1\in U_4$. Each vertex $v_{(1, j)}$ and each vertex $v_{(i, 4)}$ with $u_i\in N_H(u_1)\cap U_j$ together with the vertices in~$V_{\{u_1, u_i\}}$ (depicted as the small vertices) form a complete subgraph of size~$c$, as exem\-plarily visualized for $v_{(1, 2)}$. The color vertices $w_1$ to $w_6$ form a complete subgraph and are not depicted.}
    \label{fig:cfactor:validation}
\end{figure}

By setting $k = h^2 + \frac{1}{2}h(h-1)\cdot(c-1)$ we complete the
construction, which can clearly be performed in polynomial time and the new
parameter~$k$ solely depends on~$h$.

\emph{Correctness.}\quad $(\Rightarrow)$: Assume that graph $H$
has an $h$-colored clique $K$ and without loss of generality $K=\{u_1, u_2, \ldots,
u_h\}$ and $u_i \in U_i$ for each $i \in \{1, \ldots, h\}$. Let $F_K$ denote the set of all edges in the clique $K$. Furthermore,
let $V_K = \bigcup_{1\leq i\leq h}V_i$ and $V_{F_K} = \bigcup_{f\in F_K}V_f$. We
show that $X = W \cup V_K \cup V_{F_K}$ is a solution for
$c$-\textsc{Factor}-\ML. Note that~$|X| = h^2 +
\frac{1}{2}h(h-1)\cdot(c-1)$ by construction. It remains to show that $G_1[X]$ and $G_2[X]$ each
have $c$-factors. Observe that for any graph~$G=(V, E)$ and any partition
$\mathcal{P} = \{P_1, P_2, \ldots, P_p\}$ of~$V$ we have that, if
$G[P_i]$ has a $c$-factor for all~$i$,~$1\leq i\leq p$, then $G$ has a $c$-factor as
well.
\begin{compactenum}
  \item Note that $G_1[V_{F_K}]$ is a complete graph of order
  $\frac{1}{2}h(h-1)\cdot(c-1)\geq 2c+1$ and hence has a $c$-factor.
  Furthermore, let $u_i \in K$ and assume that $u_i \in U_j$. Then we have that $G_1[V_i \cup
  \{w_j\}]$ is by construction a $c$-regular graph and hence also has a
  $c$-factor.
  Since~$K$ is $h$-colored, $\{V_i \cup
  \{w_j\} \ | \ u_i \in K \text{ and } u_i \in U_j \}$ is a partition of~$V_K \cup W$.
  \item Note that $G_2[W]$ is a complete graph of size $h\geq c+1$ and hence
  has a $c$-factor.
  Let~$f\in F_K$ be the edge connecting $u_i$ and $u_{i'}$, and let $u_i
\in U_j$ and $u_{i'} \in U_{j'}$. Then we have that $G_2[V_f \cup \{v_{(i, j')},
v_{(i', j)}\}]$ is by construction a complete graph of order~$c+1$ and hence also
has a $c$-factor. Since $K$ is an $h$-colored clique, we have that $\{V_f \cup \{v_{(i, j')},
v_{(i', j)}\} \ | \ f \in F_K \text{ and } f = \{u_i,
u_{i'}\} \}$ is a partition of
  $V_{F_K} \cup V_K$. Notably, we also have that $\mathcal{P}_K = \{\{v_{(i, j')},
v_{(i', j)}\} \ | \ \{u_i,
u_{i'}\} = f \text{ for some } f \in F_K\}$ is a
partition of~$V_K$, since any $u_i$ is connected to $h-1$ other vertices with a
different color each and therefore we have that $\{v_{(i, j')} \ | \ \{v_{(i, j')},
v_{(i', j)}\} \in \mathcal{P}_K\} = V_i$.
\end{compactenum}
 $(\Leftarrow)$: Assume that there is a vertex set $X\subseteq V$ such that $|X|\geq k$
and $G_i[X]$ has a $c$-factor for~$\ell$ different layers $i$.
By construction of the layers we have that $G_1[X]$ and $G_2[X]$ both have
$c$-factors. Furthermore, we show the following facts:
\begin{compactdesc}
  \item[Fact 1:] \label{enum:vertices} If $v_{(i, j')} \in X$ and $u_i \in U_j$,
  then $V_i \subseteq X$ and $w_j \in X$: By construction $G_1[V_i\cup\{w_j\}]$ is a connected $c$-regular
  graph and each $v_{(i, j')}\in V_i$ is \emph{only} connected to other vertices in
  $V_i\cup\{w_j\}$ in $G_1$. Note that any proper subgraph of a connected
  $c$-regular graph is not $c$-regular and does not have a $c$-factor. It
  follows that as soon as any $v_{(i, j')}\in V_i$ is included in $X$, all other
  vertices in $V_i$ have to be included, as well as $w_j$, the vertex
  corresponding to the color of $u_i$.
  \item[Fact 2:] \label{enum:colors1} If $V_i \subseteq X$ and $V_{i'} \subseteq
  X$ and $u_i \in U_j$, then $u_{i'}\notin U_j$: Note that $G_1[V_{i'}]$ does
  not have a $c$-factor. By construction we get a connected $c$-regular graph by
  adding the vertex corresponding to the color of $u_{i'}$, hence $G_1[V_{i'}]$
  is a proper subgraph of a connected $c$-regular graph. Furthermore, we have
  that $w_j$ is already part of the $c$-regular spanning graph of
  $G_1[V_i\cup\{w_j\}]$, therefore $u_{i'}$ cannot have color $j$,
  that is, $u_{i'}\notin U_j$.
  \item[Fact 3:] \label{enum:edges2} If $X \cap V_f \neq \emptyset$ for some $f
  \in F$, then $V_f \cup \{v_{(i, j')},
v_{(i', j)}\} \subseteq X$, where $u_i \in U_j$ and~$u_{i'}\in U_{j'}$ are the endpoints of~$f$: By
construction, $G_2[V_f \cup \{v_{(i, j')}, v_{(i', j)}\}]$ is a clique of size~$c+1$ and hence a connected $c$-regular
graph. Furthermore, this clique is disconnected from the rest of $G_2$ and hence
as soon as one of its vertices is included in $X$, all of them are.
\end{compactdesc}
Now we show that $K = \{u_i \ | \ V_i \subseteq X\}$ is an $h$-colored clique in
$H$. First, we show that there must be some $v_{(i, j)}\in X$: Since $|X|\geq
h^2 + \frac{1}{2}h(h-1)\cdot(c-1)$, we have that $X \neq W$ and from Fact~3 we get that $X \not \subseteq W \cup V_F$.

By Fact~1 we know that any $V_i$ is either a
subset of $X$ or $X\cap V_i = \emptyset$.
Furthermore, if~$V_i \subseteq X$ and $u_i \in U_j$, so is $w_j$.
Fact~2 yields that we cannot have two vertices of the
same color in~$K$. This implies
\begin{equation}\label{ineq1}
|\bigcup_{V_i \subseteq X} V_i \cup W|\leq h^2.
\end{equation}
Note that $|X|\geq h^2 + \frac{1}{2}h(h-1)\cdot(c-1)$ and Inequality~\ref{ineq1}
imply that $X\cap V_F \neq \emptyset$. By Fact~3, $X$ can only include vertices corresponding to edges between vertices $u_i$
and $u_{i'}$ if $V_i \cup V_{i'} \subseteq X$. Therefore,
\begin{equation}\label{ineq2}
|X\cap V_F| \leq \frac{1}{2}h(h-1)\cdot(c-1).
\end{equation}
Note that $|X| = h^2 + \frac{1}{2}h(h-1)\cdot(c-1)$ if and only if both
Inequalities~\ref{ineq1}~and~\ref{ineq2} are equalities. This implies that $|K| = h$, all vertices in $K$ have different
colors and all colors are present (Inequality~\ref{ineq1}), and that $K$ is a clique
in $H$ (Inequality~\ref{ineq2}).

 To make
 this reduction work for any $t \geq \ell > 2$, we insert $\ell-2$
 additional layers of complete graphs and $t-\ell$ layers of edgeless graphs.
\end{proof}

\section{Hamiltonian Paths}
\label{sec:path-hard}

In this section we investigate the problem of finding Hamiltonian subgraphs, that is, subgraphs that have a simple path visiting all vertices.
\autoref{cor:hardness} in \autoref{sec:meta} states that \PathML is
\W{1}-hard when parameterized by the combined parameter~$k$ and~$\ell$. Through closer
inspection we can get a stronger result. \gprop{Hamiltonian}-\SL is known to be \NP-hard and \FPT when
parameterized by the size of the subgraph~$k$~\cite{monien1985find}.
For the multi-layer case, we can show that it is already \W{1}-hard for any
constant $\ell \geq 2$.

\begin{theorem}
\label{thm:hamiltonian}
  \PathML is \NP-hard and \W{1}-hard when parameterized by the number~$k$ of vertices to select for
  all $\ell \geq 2$ and $t \geq \ell$.
\end{theorem}
\begin{proof}
We reduce from the \textsc{Multicolored Biclique} problem. In \textsc{Multicolored Biclique} we are given a bipartite graph
$H=(U\cup W, F)$ and a partitioning that partitions vertices in $U$ and $W$ into $h$
parts each, that is, $U = U_1 \uplus \ldots \uplus U_h$ and $W = W_1 \uplus
\ldots \uplus W_h$. We call these $2h$ parts \emph{colors}, that is, each vertex of~$H$ is of exactly one color in~$\{U_1, \ldots, U_h, W_1, \ldots, W_h\}$. We need to determine whether $H$
contains a biclique of size $2h$ that contains one vertex of each color. A simple parameterized reduction from \textsc{Clique} shows that \textsc{Multicolored
  Biclique} is \W{1}-hard~\cite{DM12}.
Observe that the vertex coloring implies that any solution contains~$h$ vertices from~$U$ and~$h$ vertices from~$W$. We
will call the vertices from~$U$ \emph{low}, and the vertices from~$W$
\emph{high}, imagining that the vertices from~$W$ are on top of the vertices from~$U$. Similarly, we call colors~$U_i$ \emph{low} and colors~$W_i$~\emph{high}. Note that all colors are different.
Given an instance~$H$ of \textsc{Multicolored
  Biclique}, we construct an instance of \PathML for
  $t=\ell=2$ as follows
 and then argue that the construction is easily generalizable.
  
  \emph{Reduction idea.}\quad The main idea is to create new vertices for both the vertices of the \textsc{Multicolored
  Biclique} instance as well as the edges. We create two layers where in the first one any maximal Hamiltonian subgraph contains a path that ``selects'' one vertex from each color and one edge of each combination of a high and low color. The second layer is constructed in such a way that each maximal Hamiltonian subgraph of the first layer is also Hamiltonian in the second layer if and only if the selected edges indeed connect the selected vertices of the respective colors, implying that the selected vertices form a multicolored biclique.
  
\emph{Vertices.}\quad The vertex
set~$V$ consists of the following subsets:
\begin{compactitem}
  \item All vertices $U\cup W$ of $H$,
\item $\{s_1,s_2\}$, where we assume that~$s_1$ and~$s_2$ are not vertices
of~$H$,
\item $A_{i,j}$,~$1\le i \le h$, $1\le j\le h$,
where~$A_{i,j}:=\{\alpha_{\{u,w\}}\mid \{u,w\}\in F \wedge u\in U_i \wedge w\in
W_j\}$, and
\item $D_{i,j}$,~$1\le i\le h$, $1\le j \le h$, where
$D_{i,j}:=\{\delta_{\{w,u\}}\mid \{w,u\}\in F \wedge w\in W_i \wedge u\in
U_j\}$.
\end{compactitem}
Informally, the latter two sets are constructed by adding two vertices
for each edge of~$H$, each corresponding to one orientation of the
undirected edge. The vertices are then assigned to the vertex sets
according to the colors of their endpoints and the
orientations. Oriented edges from~$U$ to~$W$ (and their corresponding
vertices in~$V$) are called \emph{ascending}, oriented edges from~$W$ to~$U$
(and their corresponding vertices) are called \emph{descending}. Note that all these sets are non-empty unless we face a no-instance.

Now we describe how to construct a vertex and edge selection gadget by graph
$G_1$ and a validation gadget by graph $G_2$. We organize the vertices of both graphs in \emph{levels}. Each graph has~$2h^2+2h+2$
levels of vertices with edges only between neighboring levels; each~$U_i$, $W_i$,
$A_{i,j}$ and each~$D_{i,j}$ forms one level, the two remaining levels contain~$s_1$ and~$s_2$, respectively.

\emph{Vertex and edge selection gadget by
graph~$G_1$.}\quad Informally, graph~$G_1$ is constructed by putting all low
vertices and their incident ascending edges into low levels, then adding two
levels for~$s_1$ and~$s_2$, and then putting all high vertices and their incident
descending edges into high levels. More precisely,~$U_1$ is the first level
of~$G$ and~$A_{1,1}$ is the second level. Then edges are added from each~$u\in U_1$ to all vertices of~$A_{1,1}$ that correspond to an edge incident with~$u$. Then, the ascending
vertices ``incident'' with vertices from~$U_1$ are added for
increasing colors of the high endpoints. Afterwards, the vertices
from~$U_2$ are added and then the vertices corresponding to their
incident edges, and so on. The special vertices~$s_1$ and~$s_2$ are
added in two middle levels, separating the levels containing low
vertices from those containing high vertices.
Formally, the graph~$G_1$ is constructed as follows:
\begin{compactitem}
\item For each~$u\in U_i$,~$1\le i\le h$, add an edge to each~$\alpha_{\{u,w\}} \in A_{i,1}$,
\item for each~$\alpha_{\{u,w\}}\in A_{i,j}$,~$1 < j <h$, add an edge to
each~$\alpha_{\{u,w'\}}\in A_{i,j+1}$,
\item for each~$\alpha_{\{u,w\}}\in A_{i,h}$,~$1\le i <h$, add an edge to
each~$u'\in U_{i+1}$,
\item for each~$\alpha_{\{u,w\}}\in A_{h,h}$, add an edge to~$s_1$,
\item add the edge~$\{s_1,s_2\}$,
\item for each~$w\in W_{1}$, add the edge~$\{s_2,w\}$,
\item for each~$w\in W_i$,~$1 < i\le h$, add an edge to each~$\delta_{\{w,u\}} \in D_{i,1}$,
\item for each~$\delta_{\{w,u\}}\in D_{i,j}$,~$1\le j< h$, add an edge to
each~$\delta_{\{w,u'\}}\in D_{i,j+1}$, and
\item for each~$\delta_{\{w,u\}}\in D_{i,h}$,~$1\le i<h$, add an edge to
each~$w'\in W_{i+1}$.
\end{compactitem}
The idea behind the construction is that any path from the first level
to the last level corresponds to a selection of~$2h$ vertices and
of~$2h^2$ edges incident with these vertices. The vertex selection
 gadget is visualized in \autoref{fig:hamilton:selection}.

\begin{figure}[t]
\begin{center}
    \begin{tikzpicture}[line width=1pt, scale=1.5]
    \node[vert,minimum size=0.8cm] (u1) at (0,-2) {$U_1$}; 
    \node[vert,minimum size=0.8cm] (u2) at (0,-1) {$U_2$};
    \node[vert,minimum size=0.8cm] (u3) at (0,0) {$U_3$};
    \node[vert,minimum size=0.8cm] (u4) at (0,1) {$U_4$};
    \node[vert,minimum size=0.8cm] (u5) at (0,2) {$U_5$};
    \node[vert,minimum size=0.8cm] (a11) at (1,-2) {$A_{1,1}$}; 
    \node[vert,minimum size=0.8cm] (a21) at (1,-1) {$A_{2,1}$};
    \node[vert,minimum size=0.8cm] (a31) at (1,0) {$A_{3,1}$};
    \node[vert,minimum size=0.8cm] (a41) at (1,1) {$A_{4,1}$};
    \node[vert,minimum size=0.8cm] (a51) at (1,2) {$A_{5,1}$};
    \node[vert,minimum size=0.8cm] (a12) at (2,-2) {$A_{1,2}$}; 
    \node[vert,minimum size=0.8cm] (a22) at (2,-1) {$A_{2,2}$};
    \node[vert,minimum size=0.8cm] (a32) at (2,0) {$A_{3,2}$};
    \node[vert,minimum size=0.8cm] (a42) at (2,1) {$A_{4,2}$};
    \node[vert,minimum size=0.8cm] (a52) at (2,2) {$A_{5,2}$};
    \node[vert,minimum size=0.8cm] (a13) at (3,-2) {$A_{1,3}$}; 
    \node[vert,minimum size=0.8cm] (a23) at (3,-1) {$A_{2,3}$};
    \node[vert,minimum size=0.8cm] (a33) at (3,0) {$A_{3,3}$};
    \node[vert,minimum size=0.8cm] (a43) at (3,1) {$A_{4,3}$};
    \node[vert,minimum size=0.8cm] (a53) at (3,2) {$A_{5,3}$};
    \node[vert,minimum size=0.8cm] (a14) at (4,-2) {$A_{1,4}$}; 
    \node[vert,minimum size=0.8cm] (a24) at (4,-1) {$A_{2,4}$};
    \node[vert,minimum size=0.8cm] (a34) at (4,0) {$A_{3,4}$};
    \node[vert,minimum size=0.8cm] (a44) at (4,1) {$A_{4,4}$};
    \node[vert,minimum size=0.8cm] (a54) at (4,2) {$A_{5,4}$};
    \node[vert,minimum size=0.8cm] (a15) at (5,-2) {$A_{1,5}$}; 
    \node[vert,minimum size=0.8cm] (a25) at (5,-1) {$A_{2,5}$};
    \node[vert,minimum size=0.8cm] (a35) at (5,0) {$A_{3,5}$};
    \node[vert,minimum size=0.8cm] (a45) at (5,1) {$A_{4,5}$};
    \node[vert,minimum size=0.8cm] (a55) at (5,2) {$A_{5,5}$};
    
    \node[vert] (s1) at (2.5,3) {$s_1$};
    \node[vert] (s2) at (2.5,3.75) {$s_2$};

    \draw (s1) -- (s2);

    \begin{pgfonlayer}{bg}
    \path (u1.center) edge[color=lightgray,line width=6pt] (a11.center);
    \path (a11.center) edge[color=lightgray,line width=6pt] (a12.center);
    \path (a12.center) edge[color=lightgray,line width=6pt] (a13.center);
    \path (a13.center) edge[color=lightgray,line width=6pt] (a14.center);
    \path (a14.center) edge[color=lightgray,line width=6pt] (a15.center);
    \path (u2.center) edge[color=lightgray,line width=6pt] (a21.center);
    \path (a21.center) edge[color=lightgray,line width=6pt] (a22.center);
    \path (a22.center) edge[color=lightgray,line width=6pt] (a23.center);
    \path (a23.center) edge[color=lightgray,line width=6pt] (a24.center);
    \path (a24.center) edge[color=lightgray,line width=6pt] (a25.center);
    \path (u3.center) edge[color=lightgray,line width=6pt] (a31.center);
    \path (a31.center) edge[color=lightgray,line width=6pt] (a32.center);
    \path (a32.center) edge[color=lightgray,line width=6pt] (a33.center);
    \path (a33.center) edge[color=lightgray,line width=6pt] (a34.center);
    \path (a34.center) edge[color=lightgray,line width=6pt] (a35.center);
    \path (u4.center) edge[color=lightgray,line width=6pt] (a41.center);
    \path (a41.center) edge[color=lightgray,line width=6pt] (a42.center);
    \path (a42.center) edge[color=lightgray,line width=6pt] (a43.center);
    \path (a43.center) edge[color=lightgray,line width=6pt] (a44.center);
    \path (u5.center) edge[color=lightgray,line width=6pt] (a51.center);
    \path (a51.center) edge[color=lightgray,line width=6pt] (a52.center);
    \path (a52.center) edge[color=lightgray,line width=6pt] (a53.center);
    \path (a53.center) edge[color=lightgray,line width=6pt] (a54.center);
    \path (a54.center) edge[color=lightgray,line width=6pt] (a55.center);
    \path (a44.center) edge[color=lightgray,line width=6pt] (a45.center);
   \path (a15.center) edge[color=black,line width=6pt, out=150, in=-30] (u2.center);
      \path (a25.center) edge[color=black,line width=6pt, out=150, in=-30] (u3.center);
      \path (a35.center) edge[color=black,line width=6pt, out=150, in=-30] (u4.center);
      \path (a45.center) edge[color=black,line width=6pt, out=150, in=-30] (u5.center);
      \path (a55.center) edge[line width=6pt, out=150, in=-30] (s1.center);
    \end{pgfonlayer}

    \end{tikzpicture}
    \end{center}
    \caption{The lower half of the vertex selection gadget $G_1$ for $h=5$. The circles annotated with $U_i$ or $A_{i, j}$ correspond to the respective vertex set from the construction, $s_1$ and $s_2$ are single vertices connected by an edge. The thick gray edges represent the edges between the two vertex sets they connect. The thick black edges indicate that all edges between the two vertex sets are present. The upper half of the gadget is symmetric and not depicted.}
    \label{fig:hamilton:selection}
\end{figure}

\emph{Validation gadget by graph~$G_2$.}\quad With the
second graph~$G_2$ we enforce that the selected ascending and descending
edges between each color pair~$U_i$ and~$W_j$ correspond to the same edge
in~$H$ and that any path of length~$2h+2h^2+1$ passes through each level
of~$G_1$ and each level of~$G_2$.
Formally,~$G_2$ is constructed as follows. Herein, assume an arbitrary
but fixed ordering on pairs of low and high colors.
\begin{compactitem}
\item Level~1 contains~$s_1$,
\item level~$2i$,~$1\le i\le h^2$, contains all ascending vertices of the~$i$th
  color pair,
\item level~$2i+1$,~$1\le i\le h^2$, contains all descending vertices of
  the~$i$th color pair,
\item level~$2h^2+2$ contains~$s_2$,
\item level~$2h^2+2+i$,~$1\le i\le h$, contains all vertices from~$U_i$, and
\item level~$2h^2+h+2+i$,~$1\le i\le h$, contains all vertices from~$W_i$.
\end{compactitem}
All edges between consecutive levels are added except for the
levels~$2i$ and~$2i+1$,~$1\le i\le h^2$: Here, we add only an edge between
vertices that correspond to the same edge, that is, we add the edge
set~$\{\alpha_{\{u,w\}},\delta_{\{w,u\}}\mid \{u,w\}\in F\}.$ The validation
 gadget is visualized in \autoref{fig:hamilton:validation}.

\begin{figure}[t]
\begin{center}
    \begin{tikzpicture}[line width=1pt, scale=1.5]
    \node[vert,minimum size=0.8cm] (u2) at (-.5,.5) {$D_{5, 5}$};
    \node[vert,minimum size=0.8cm] (u4) at (-.5,-1) {$D_{2, 5}$};
    \node[vert,minimum size=0.8cm] (u5) at (-.5,-2) {$D_{1,5}$};
    \node[vert,minimum size=0.8cm] (a21) at (.5,.5) {$A_{5,5}$};
    \node[vert,minimum size=0.8cm] (a41) at (.5,-1) {$A_{2,5}$};
    \node[vert,minimum size=0.8cm] (a51) at (0.5,-2) {$A_{1,5}$};
    \node[vert,minimum size=0.8cm] (a52) at (2,-2) {$D_{1,2}$};
    \node[vert,minimum size=0.8cm] (a53) at (3,-2) {$A_{1,2}$};
    \node[vert,minimum size=0.8cm] (a24) at (4,.5) {$D_{5,1}$};
    \node[vert,minimum size=0.8cm] (a44) at (4,-1) {$D_{2,1}$};
    \node[vert,minimum size=0.8cm] (a54) at (4,-2) {$D_{1,1}$};
    \node[vert,minimum size=0.8cm] (a25) at (5,.5) {$A_{5,1}$};
    \node[vert,minimum size=0.8cm] (a45) at (5,-1) {$A_{2,1}$};
    \node[vert,minimum size=0.8cm] (a55) at (5,-2) {$A_{1,1}$};
    
    \node[vert] (s1) at (2.25,-2.75) {$s_1$};
    \node[vert] (s2) at (2.25,1.25) {$s_2$};
 
    \begin{pgfonlayer}{bg}
      
    \path (s2.center) edge[line width=6pt, out=-150, in=40] (u2.center);
    \path (u2.center) edge[color=lightgray,line width=6pt] (a21.center);
    \path (a21.center) edge[color=lightgray,line width=6pt, dashed] (a24.center);
     \path (a24.center) edge[color=lightgray,line width=6pt] (a25.center);
    \path (a25.center) edge[color=lightgray,line width=6pt, out=-150, in=30, dashed] (u4.center);
    \path (u4.center) edge[color=lightgray,line width=6pt] (a41.center);
    \path (a41.center) edge[color=lightgray,line width=6pt, dashed] (a44.center);
    \path (a44.center) edge[color=lightgray,line width=6pt] (a45.center);
    \path (a45.center) edge[line width=6pt, out=-150, in=30] (u5.center);
    \path (u5.center) edge[color=lightgray,line width=6pt] (a51.center);
    \path (a51.center) edge[color=lightgray,line width=6pt, dashed] (a52.center);
    \path (a52.center) edge[color=lightgray,line width=6pt] (a53.center);
    \path (a53.center) edge[color=lightgray,line width=6pt] (a54.center);
    \path (a54.center) edge[color=lightgray,line width=6pt] (a55.center);
    \path (a55.center) edge[line width=6pt, out=-140, in=30] (s1.center);
    \end{pgfonlayer}
    

    \end{tikzpicture}
    \end{center}
    \caption{The lower half of the levels of the validation gadget $G_2$ for $h=5$. The circles annotated with $A_{i, j}$ or $D_{i, j}$ correspond to the respective vertex set from the construction, $s_1$ and $s_2$ are single vertices. The thick gray edges represent the edges between the two vertex sets they represent. The thick black edges indicate that all possible edges between the respective vertex sets are present. Dashed edges indicate that some vertex sets are not visualized. The higher levels of the gadget are not depicted.}
    \label{fig:hamilton:validation}
\end{figure}

To
finish the construction we set~$k:=2h+2h^2+2$.
 The reduction clearly runs in polynomial time and the new
  parameter~$k$ depends only on the parameter~$h$ of the
  \textsc{Multicolored Biclique} instance. 
  
  \emph{Correctness.}\quad Thus it remains to show
  equivalence of the instances. 
  
  $(\Rightarrow)$:
  Let~$K=\{u_1,\ldots ,u_h, w_1, \ldots, w_h\}$ be a multi-colored biclique in $H$.
  Let~$X$ be the vertex set containing~$K$,~$s_1$, and~$s_2$, and for each edge~$e$
  of~$H[K]$, the ascending and the descending vertex corresponding
  to~$e$. We show that~$G_1[X]$ and~$G_2[X]$ have a Hamiltonian
  path. In~$G_1$, this path starts at~$u_1$, then
  visits~$\alpha_{u_1,w_{1}}$, then~$\alpha_{u_1,w_{2}}$
  until~$\alpha_{u_1,w_{h}}$. Then it visits~$u_2$ and the ascending
   vertices corresponding to edges incident with~$u_2$ in the same
  fashion, that is, first~$\alpha_{u_2,w_{1}}$,
  then~$\alpha_{u_2,w_{2}}$, and so on. This is continued
  until~$\alpha_{u_h,w_{h}}$ is visited. Then, the path visits~$s_1$
  and~$s_2$. Then it visits the high vertices and the descending
  vertices for their incident edges in the same fashion. By
  construction, all necessary edges are present: neighboring edge
  vertices correspond to edges that share one endpoint, and after a
  vertex~$\alpha_i$, the next visited edge vertex is incident
  with~$\alpha_i$.
   
  In~$G_2$, the path visits each level exactly once, going from
  level~1 through level~$k$. The necessary edges are present since
  the only neighboring levels that are not complete bipartite graphs
  are those that contain ascending and descending vertices of the same
  color pair. Since for each color pair the ascending and descending
   vertex in~$S$ correspond to the same edge in~$H$, they are
  adjacent in~$G_2$.

  \looseness -1
 $(\Leftarrow)$: Observe that any vertex set consisting only of
  vertices from~$\{s_1,s_2\}\cup \bigcup_{1\le i\le h} (U_i \cup W_i)$ is
  disconnected either in~$G_1$ or in~$G_2$. Thus, the set~$X$ contains
  at least one ascending or descending vertex. In either case it
  must also contain an edge vertex of the other type: In~$G_2$, the
  levels containing ascending and descending vertices alternate and
  any set containing either only vertices of the first two
  levels of~$G_2$ or only of vertices of the last~$2h+2$ levels is disconnected
  in~$G_1$. Now, since~$X$ contains an ascending and a descending
  vertex and since~$G_1[X]$ is connected, the vertices~$s_1$ and~$s_2$
  are contained in~$X$. Thus, the vertex set~$X$ contains a vertex
  from the first and the last level of~$2h+2h^2+2$ levels
  in~$G_2$. This implies that it contains exactly one vertex of each
  level of~$G_2$. Thus, the vertex set~$X$ contains exactly~$h$ low
  vertices and~$h$ high vertices with different colors and it contains
  for each color pair an ascending and a descending vertex. By
  construction of~$G_2$ and since~$G_2[X]$ has a Hamiltonian path
  visiting each level exactly once, these two vertices are the same,
  that is, the~$2h^2$ ascending and descending vertices correspond
  to~$h^2$ edges in~$H$. By construction of~$G_1$ and since~$G_1[X]$
  has a Hamiltonian path visiting each level of~$G_1$ exactly once,
  these edges are incident only with vertices of~$(U\cup W)\cap
  X$. Thus,~$H[(U\cup W)\cap X]$ is a multicolored biclique. 
  
 To make
 this reduction work for any $t \geq \ell \geq 2$, we can insert $\ell-2$
 additional layers of complete graphs and $t-\ell$ layers of edgeless graphs.
\end{proof}

\section{Conclusion}
\label{sec:conclusion}
We performed a systematic study of the (parameterized) computational complexity of subgraph detection problems in
multi-layer networks. In particular, we encountered several computational hardness results
for multi-layer subgraph detection problems that are solvable in
polynomial time in the single-layer case. 
In the following, we list
some possibilities for future research with the goal to obtain positive algorithmic results.

First, the case of two-layer graphs should continue to receive special
attention. We showed that \MatchML is solvable in polynomial time in
the two-layer case, whereas it is \W{1}-hard when parameterized by the
number of vertices to select~$k$ for three or more layers. Considering
acyclic subgraphs, \citet{ALMS15} also
showed specialized algorithms for the two-layer case. Also \citet{CY14} focused mostly on the two-layer case. It
would be interesting to systematically explore which subgraph
detection problems are tractable in the two-layer case and to
identify more problems that behave differently for two and three
layers.

Second, in many applications the input graphs are directed. One of
our hardness results transfers directly to this case: The construction
in the reduction from \textsc{Multicolored Biclique} to
\PathML (proof of Theorem~\ref{thm:hamiltonian}) can
be easily adapted to yield directed acyclic graphs by orienting all
edges from lower levels to higher levels, implying that \PathML is \NP-hard and \W{1}-hard when parameterized by~$k$ for
  all $\ell \geq 2$ and $t \geq \ell$ if every layer is a directed acyclic
  graph.
Hence, for directed acyclic graphs the complexity gap between the
cases with one and two layers is even bigger because we can find a longest path in a directed acyclic graph 
in polynomial-time (as opposed to being \NP-hard and \FPT with respect to
subgraph order~$k$ in the undirected single-layer case). As also already mentioned by \citet{CY14}, finding positive algorithmic results for multi-layer subgraph problems in directed graphs is a challenging open direction.

Finally, in some applications, one is interested in different subgraph
properties for each layer~\cite{CY14,agrawal_et_al18}. For example, in
the \textsc{Supported Path} problem, one layer is undirected and one
layer is directed, and one aims to find an induced subgraph that is
connected in the undirected layer and Hamiltonian in the directed
layer~\cite{FKMR15}. There is only limited systematic
investigation of the complexity of such mixed multi-layer subgraph
problems~\cite{CY14,agrawal_et_al18}.

\paragraph{Acknowledgment.}
\begin{wrapfigure}{r}{0.17\textwidth}\vspace{-3mm}
\includegraphics[width=.8\linewidth]{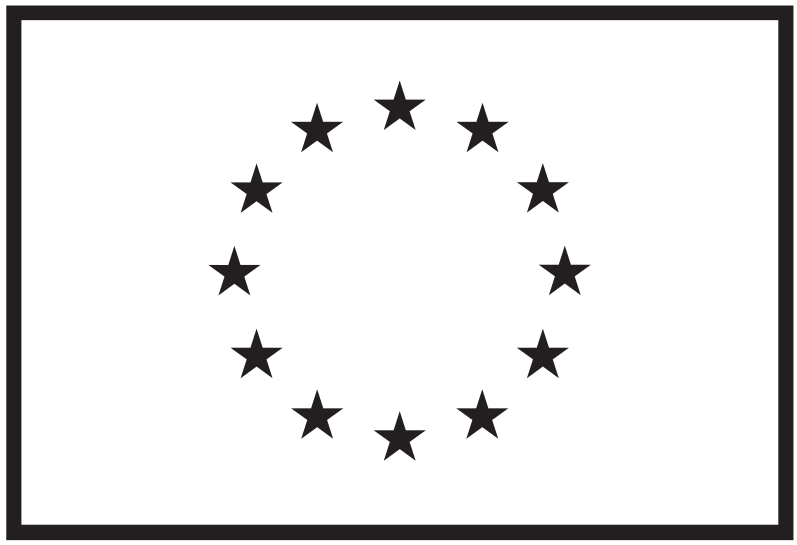}\\[3mm]
\includegraphics[width=.8\linewidth]{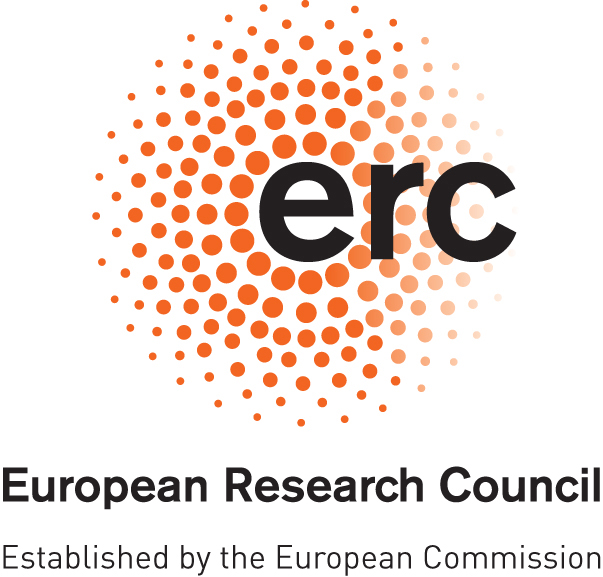}
\end{wrapfigure}
HM was partially supported by the DFG, projects DAPA (NI~369/12) and MATE (NI~369/17). MS was supported by the DFG, project DAPA (NI~369/12), the People Programme (Marie Curie Actions) of the European Union's Seventh Framework Programme (FP7/2007-2013) under REA grant agreement number {631163.11},
by the Israel Science Foundation (grant no. 551145/14), and by
the European Research Council (ERC) under the European Union's Horizon 2020 research and innovation programme under grant agreement number {714704}.
Parts of MS' work were done while with TU Berlin, Germany and Ben-Gurion University of the Negev, Beer Sheva, Israel. CK was supported by the DFG, project MAGZ (KO 3669/4-1).
RB was
partially supported by the DFG, fellowship BR 5207/2.
This work was initiated at the research retreat of the TU~Berlin Algorithmics and Computational Complexity research group held in Darlingerode, Harz mountains, Germany, April~2014. The authors would like to thank Sepp Hartung for initial
discussions.

\bibliographystyle{abbrvnat}
\bibliography{bibliography} 

\appendix
\section{Definitions}
\label{appendix:definitions}
\begin{description}
\item[Asteroidal Triple-Free Graph] An independent set of size three where each pair of vertices is joined by a path that
avoids the neighborhood of the third is called an \emph{asteroidal triple}. A graph is \emph{asteroidal triple-free} if it does not contain asteroidal triples.
\item[$c$-Colorable Graph] A graph is \emph{$c$-colorable} if there is a way of coloring the vertices with at most $c$ different colors such that no two adjacent vertices share the same color.\
\item[$c$-Connectivity] A graph is called \emph{$c$-connected} if it contains at least $c+1$ vertices, but does not contain a set of $c-1$ vertices whose removal disconnects the graph.
\item[$c$-Core] A graph is called a \emph{$c$-core} if each vertex has degree at least $c$.
\item[$c$-Edge-Connectivity] A graph is called \emph{$c$-edge-connected} if it does not contain a set of $c-1$ edges whose removal disconnects the graph.
\item[$c$-Factor] A graph has a \emph{$c$-factor} if it has a $c$-regular spanning graph.
\item[$c$-Regular Graph] A graph is called \emph{$c$-regular} if every vertex has degree $c$.
\item[$c$-Truss] A graph is called a \emph{$c$-truss} if it is connected and each edge is contained in at least $c-2$ triangles.
\item[Chordal Graph] A graph is called \emph{chordal} if each induced cycle has at most three vertices.
\item[Cluster Graph] A graph is called a \emph{cluster graph} if it is a collection of disjoint cliques.
\item[Cograph] A graph is called a \emph{cograph} if it does not contain any induced path of length four.
\item[Comparability Graph] A graph is called a \emph{comparability graph} if there is a partial order over the vertices such that each pair of vertices is adjacent if and only if it is comparable.
\item[Complete Multipartite Graph] A graph is called \emph{complete multipartite} if the vertex set can be partitioned such that each vertex pair is adjacent if and only if the two vertices are in different partitions.
\item[Edgeless Graph] A graph is called \emph{edgeless} is it does not contain any edges.
\item[Forest] A graph is called a \emph{forest} if it is a collection of trees.
\item[Hamiltonian Graph] A graph is called \emph{Hamiltonian} if it contains a \emph{Hamiltonian path}, that is, a simple path that visits each vertex exactly once.
\item[$h$-Index] The \emph{$h$-index} of a graph is the largest integer $h$ such that the graph contains at least $h$ vertices with degree at least $h$.
\item[Interval Graph] A graph is called an \emph{interval graph} if an interval of the real numbers can be assigned to each vertex such that two vertices are adjacent if and only if the intervals overlap.
\item[Line Graph] A graph is called a \emph{line graph} if there is another graph such that each vertex of the line graph corresponds to an edge of the other graph and two vertices in the line graph are adjacent if the corresponding edges in the other graph share a common endpoint.
\item[Matching] A graph has a \emph{matching} if it has a $1$-factor.
  \item[Perfect Graph] A graph is called \emph{perfect} if the chromatic number of every induced subgraph equals the size of the largest clique of that subgraph. The \emph{chromatic number} of a graph is the smallest $c$ such that the graph is $c$-colorable.
\item[Permutation Graph] A graph is called a \emph{permutation graph} if the vertices represent elements in a permutation such that two vertices are adjacent if and only if the respective pair of elements is reversed by the permutation.
\item[Planar Graph] A graph is called \emph{planar} if it can be embedded in the plane without edge-crossings, that is, it can be drawn on the plane such that its edges only intersect at their endpoints.
\item[Quasi-Threshold Graph] A graph is called a \emph{quasi-threshold graph} if it does not contain any induced path of length four and any induced cycle of length four.
\item[Split Graph] A graph is called a \emph{split graph} if its vertices can be partitioned into a clique and an independent set.
\item[Star] A graph is called a \emph{star} if it is a tree with only one internal node.
\end{description}
\end{document}